\documentclass[10pt,journal,compsoc]{IEEEtran}
\usepackage[T1]{fontenc}
\usepackage[latin9]{inputenc}
\usepackage{float}
\usepackage{amstext}
\usepackage{amsthm}
\usepackage{amssymb}
\usepackage{amsmath}
\usepackage{graphicx}

\makeatletter

\providecommand{\tabularnewline}{\\}
\floatstyle{ruled}
\newfloat{algorithm}{tbp}{loa}
\providecommand{\algorithmname}{Algorithm}
\floatname{algorithm}{\protect\algorithmname}

\newfloat{protocol}{tbp}{loa}
\providecommand{\protocolname}{Protocol}
\floatname{protocol}{\protect\protocolname}

\theoremstyle{plain}

\theoremstyle{definition}
\newtheorem{defn}{\protect\definitionname}
\theoremstyle{plain}
\newtheorem{prop}{\protect\propositionname}
\newtheorem{corollary}{Corollary}
\newtheorem{example}{Example}
\usepackage[caption=false,font=footnotesize]{subfig}

\makeatother

\providecommand{\definitionname}{Definition}
\providecommand{\propositionname}{Proposition}
\providecommand{\theoremname}{Theorem}
\providecommand{\protocolname}{Protocol}

\begin{document}
%Jordi2020. \title{High-Dimensional Randomized Response}
\title{Multi-Dimensional Randomized Response}

% author names and IEEE memberships
% note positions of commas and nonbreaking spaces ( ~ ) LaTeX will not break
% a structure at a ~ so this keeps an author's name from being broken across
% two lines.
% use \thanks{} to gain access to the first footnote area
% a separate \thanks must be used for each paragraph as LaTeX2e's \thanks
% was not built to handle multiple paragraphs
%

\author{Josep Domingo-Ferrer,~\IEEEmembership{Fellow, IEEE,}
	and Jordi Soria-Comas% <-this % stops a space
	\thanks{J. Domingo-Ferrer is with Universitat Rovira i Virgili,
Dept. of Computer Engineering and Mathematics, UNESCO Chair in Data Privacy,
CYBERCAT-Center for Cybersecurity Research of Catalonia,
		 Av. Pa\"{\i}sos Catalans 26,
		43007 Tarragona, Catalonia. E-mail josep.domingo@urv.cat 

J. Soria-Comas is with the Catalan Data Protection Authority, Rossell\'o 214, 
Esc. A, 1r 1a, 08008 Barcelona, Catalonia. E-mail jordi.soria@gencat.cat}}%	
	
\markboth{~Vol.~?, No.~?, Month~YYYY}%
{Josep Domingo-Ferrer \MakeLowercase{\textit{et al.}}: High-Dimensional Randomized Response}

\IEEEtitleabstractindextext{%
	\begin{abstract}
		In our data world, a host of not necessarily trusted controllers
		gather data on individual subjects. To preserve her privacy
		and, more generally, her informational self-determination,
		the individual has to be empowered by giving her agency on 
		her own data. Maximum agency is afforded by 
		local anonymization, that allows each individual
		to anonymize her own data before handing them to the data controller.
		Randomized response (RR) is a local anonymization approach able to yield multi-dimensional full sets of anonymized microdata
		that are valid for exploratory analysis and machine learning. This is so because
		an unbiased estimate of the distribution of the true
		data of individuals can be obtained from their pooled randomized 
		data. Furthermore, RR offers 
		rigorous privacy guarantees.
		The main weakness of RR is the curse of dimensionality when
		applied to several attributes: as the number of attributes grows,
		the accuracy of the estimated true data distribution quickly degrades.
		We propose several complementary approaches to mitigate  
		the dimensionality problem. First, we present two basic protocols,
		separate RR on each attribute and joint RR for all attributes,
		and discuss their limitations. Then we introduce an algorithm
		to form clusters of attributes so that attributes
		in different clusters can be viewed as independent and joint
		RR can be performed within each cluster. After that, 
		we introduce an adjustment algorithm for the randomized data set
		that repairs some of the accuracy loss due to 
		assuming independence between attributes when using RR separately
		on each attribute or due to assuming independence between clusters in cluster-wise RR.
		We also present empirical work to illustrate the proposed methods.
	\end{abstract}
	% Note that keywords are not normally used for peerreview papers.
	}

% make the title area
\maketitle

% For peer review papers, you can put extra information on the cover
% page as needed:
% \ifCLASSOPTIONpeerreview
% \begin{center} \bfseries EDICS Category: 3-BBND \end{center}
% \fi
%
% For peerreview papers, this IEEEtran command inserts a page break and
% creates the second title. It will be ignored for other modes.
\IEEEpeerreviewmaketitle

\section{Introduction}

Twenty years ago, National Statistical Institutes 
and a few others were the only data controllers explicitly gathering 
data on citizens, and their legal status often made them trusted. 
In contrast, in the current big data scenario, there is a host 
of controllers gathering information, and it is no longer reasonable 
to take it for granted that the individual 
subject trusts all of them to keep her data confidential 
and/or to anonymize them properly in case of release~\cite{adg55}.

Thus, to preserve her privacy and, more generally, her informational self-determination, the individual has to be empowered by giving her agency 
on her own data. Local anonymization is a paradigm
in which each individual anonymizes her data before
handing them to the data controller, thereby giving 
maximum agency to the individual. 
Several masking methods coming from 
 statistical disclosure control (SDC~\cite{Hundepool}) can be applied
locally, including generalization/recoding and noise addition. 
On the other hand, there are methods specifically
designed for local anonymization that, in addition to helping subjects
hide their responses, allow the data controller to get an accurate estimation
of the distribution of responses for groups of subjects (for example,
randomized response~\cite{citerr1965,citerr1969} and FRAPP~\cite{agrawal2005}).
Also, a number of methods have been proposed to obtain differentially
private (DP)
 preselected statistics via local anonymization, like
RAPPOR~\cite{rappor} and local 
DP~\cite{cormodemarginal2018,cormodetutorial2018}.
%Deixo fora Wang 2019.

While most of the local anonymization 
approaches were designed to obtain 
statistics on the set of individuals
who contribute their locally anonymized input, 
randomized response (RR) has the attractive feature 
of being able to output multi-dimensional full sets of 
anonymized microdata (individual records) 
that are valid for exploratory analysis.
Indeed, an unbiased estimate of  
the distribution of the original microdata (corresponding
to the true attribute values of individuals) can be obtained 
from the empirical distribution of the released 
randomized microdata.
Most exploratory analyses
and statistical calculations
can be performed based on this estimated distribution, 
including what~\cite{psd2018} call multi-party computation with 
statistical input confidentiality. 
It is even possible to re-create 
a synthetic estimate of the original data set 
by repeating each combination of attribute values as many times
as dictated by its frequency in the estimated joint distribution.

Furthermore, the privacy guarantees afforded by RR are 
easily expressible in terms of rigorous privacy models such 
as differential privacy~\cite{Wang16,Wang14} or 
information-theoretic secrecy~\cite{secrypt}.

Unfortunately, the picture is not as rosy as suggested by 
the two previous paragraphs. Like so many methods, randomized response
suffers from the curse of dimensionality:
\begin{itemize}
\item Applying RR simultaneously to a set of attributes
amounts to applying it to the Cartesian product of those
attributes, which has a number of possible categories that
grows exponentially with the number of attributes.
Unless the number of individuals providing input is much
greater than the number of categories, the accuracy
and hence the utility of the estimated distribution 
of the original data will be poor.
\item  RR can certainly 
be applied separately to each single attribute. 
 %Jordi2020. However, if it is separately applied to two or more attributes
%that are dependent on each other,that dependence will be 
%attenuated or even lost in the randomized data. 
%JOSEP2020. Revisat.
Nevertheless, by doing so the ability to estimate the joint distribution of
the original data based on the randomized data is lost; 
 only the 
marginal distributions of attributes can be estimated.
This also entails a loss of accuracy and hence a loss of utility 
of the estimated distributions.
\end{itemize}

%Jordi2020
Dimensionality issues are common to all data anonymization techniques. However, 
effectively dealing with them is tougher in the local anonymization 
paradigm because the global 
picture of the data is missing. Research in this topic is ample. 
In particular, in differential privacy and local differential privacy, 
many strategies have been proposed to cope with a high number of attributes,
%HOSEP2020. He mirat de separar les referències segons el que feien
%servir, perquè em semblava que anonemar 3 principis i després
%donar totes les referències juntes no quedava bé. No sé si he
%fet correspondre bé cada paper amb cada principi (me'ls he mirats
%per damunt).
 such as dealing with $k$-way marginals~\cite{calm,priview,privbayes}, 
taking advantage of sparsity~\cite{privhd,cormodemarginal2018} 
or taking advantage of the dependence between 
attributes~\cite{LDPfreq,extremal}.

\subsection*{Contribution and plan of this paper}

In this paper, we propose several complementary approaches 
to mitigate the curse of dimensionality
in randomized response. The paper's contributions are as follows:
\begin{itemize}
\item We first present two basic protocols for RR and discuss
their limitations: one performs separate RR for each attribute
and the other joint RR for all attributes.
\item We then describe an intermediate approach based
on identifying clusters of attributes such that attributes in different
clusters can be viewed as (nearly) independent. In this way,
joint RR can be performed for the attributes in each cluster.
Since our clustering algorithm requires as input the dependences between
attributes, we present several methods for assessing these dependences
in an RR scenario in which the true data of each individual 
must stay confidential to that individual.
\item After that, we introduce an adjustment algorithm of the
randomized data set that ``repairs'' some of the accuracy loss
incurred by assuming independence between attributes
when using RR separately on each attribute or assuming independence
between clusters when using RR cluster-wise.
\end{itemize}

Section~\ref{back} gives background  
 on RR, its privacy
guarantees and its estimation error.
Section~\ref{sec:basic} introduces the two basic RR protocols:
separate RR for each attribute and joint RR for all attributes.
Section~\ref{sec:RR-clusters} presents the approach
based on attribute clustering and the methods for 
privacy-preserving evaluation of attribute dependences. 
Section~\ref{sec:Adjustment} details 
the adjustment to reduce the accuracy loss caused by 
the independence assumptions.
Experimental work is reported in Section~\ref{sec:results}.
Section~\ref{related} reviews related work.
Finally, conclusions and future research directions are
gathered in Section~\ref{conc}.

\section{Background}
\label{back}

\subsection{Randomized response}
\label{sub:back_RR}

Randomized response~\cite{citerr1965,citerr1969} 
is a mechanism that respondents 
to a survey can use to protect their privacy 
when asked about the value of sensitive attribute
({\em e.g.} did you take drugs last month?).
In many respects RR was a forerunner when proposed in the 1960s:
it was not only an anonymization method
{\em avant la lettre} (before anonymization
and statistical disclosure control were introduced 
by Dalenius~\cite{Dalenius} a decade later), but it 
ushered in the even more modern notion of {\em local}
anonymization. 
Closely related to RR are 
%Jordi3.
%the subsequent post-randomization method (PRAM).
%Both methods differ on who performs the randomization~\cite{VandenHout}:
%JOSEP4. Canviat una mica i incloses més referències.
the more recent PRAM~\cite{citepram} and FRAPP~\cite{agrawal2005} methods. 
PRAM, which stands for post-randomization method, 
differs from RR on who performs the randomization~\cite{VandenHout}:
whereas in RR it is the individual before delivering her
response, in PRAM it is the data controller after collecting
all responses (hence the name post-randomization).
%Jordi3.
FRAPP extends the original RR method along the lines 
proposed in~\cite{Chaud88}.

Beyond historical merit, a strong
point of RR is that the data collector can still 
estimate from the randomized responses 
the proportion of each of the possible {\em true} answers
of the respondents.

Let us denote by $X$ the attribute containing the answer
to the sensitive question. If $X$ can take $r$ possible
values, then 
the randomized response $Y$ reported 
by the respondent instead of $X$ follows an $r \times r$ 
matrix of probabilities
\begin{equation}
\label{probmatrix}
{\bf P} = \left( \begin{array}{ccc} p_{11} & \cdots & p_{1r} \\
\vdots & \vdots & \vdots \\
p_{r1} & \cdots & p_{rr} \end{array}\right),
\end{equation}
where $p_{uv} = \Pr(Y = v| X=u)$, for $u,v \in \{1,\ldots,r\}$
denotes the probability that the randomized response 
is $v$ when the respondent's true attribute value is $u$.

Let $\pi_1,\ldots,\pi_r$ be the proportions of respondents whose true values 
fall in each of the $r$ categories of $X$ and let 
$\lambda_v = \sum_{u=1}^r p_{uv} \pi_u$ for $v=1,\ldots, r$, 
be the probability of the reported value $Y$ being $v$.
If we define ${\bf \lambda} = (\lambda_1,\ldots,\lambda_r)^T$ 
and ${\bf \pi} = (\pi_1,\ldots, \pi_r)^T$, it holds that
${\bf \lambda} = {\bf P}^T {\bf \pi}$. Furthermore, 
if $\hat{\bf \lambda}$ is the vector of sample proportions
corresponding to ${\bf \lambda}$ and ${\bf P}$ is nonsingular,
in Chapter 3.3 
of~\cite{Chaud88} it is proven that an unbiased estimator
${\bf \pi}$ can be computed as
\begin{equation}
\label{unbiased}
\hat{\bf \pi}  = ({\bf P}^T)^{-1} \hat{\bf \lambda}
\end{equation}
and an unbiased estimator of the dispersion matrix is also provided.
In particular, 
the larger the off-diagonal probability mass in ${\bf P}$, the more
dispersion (and the more respondent protection).

The estimation obtained from Equation (\ref{unbiased}) may
not be a proper probability distribution: it may have
values below 0 and above 1. This happens when the empirical
distribution of the randomized data is not consistent with 
the randomization matrix. For instance, if all the values
in the first column of ${\bf P}$ are greater than 0.5, then we 
should expect the frequency of the first category to be greater
than 0.5. If it is not, then Equation (\ref{unbiased}) will
necessarily return some negative values. 
In~\cite{Alvim2018} an iterative Bayesian update is proposed that converges 
to a proper probability distribution. In 
Section~\ref{truedist}, we describe a simpler 
solution to ensure a proper distribution. 

\subsection{Privacy guarantees}
\label{sec3}
The confidentiality guarantee given by RR 
results from each individual potentially altering her response by 
randomly drawing from a previously fixed distribution. Thus, given the
individual's randomized response, 
we are uncertain about what her true response would have been. 

In spite of the previous intrinsic guarantee of randomized
response 
%Jordi3.
%, it will be useful in the sequel to quantify the privacy afforded by 
%randomized response in terms of differential privacy.
and given the popularity of differential privacy~\cite{Dwork2006}, we will also
quantify the privacy afforded by 
randomized response in terms of differential privacy.
However, we would like to remark that {\em attaining a given level of 
differential privacy is not the goal of this work}. Among the algorithms
 we propose, some enforce differential privacy while others only
qualify as differentially private if certain assumptions 
are made about the information
that is publicly available. While we do not claim that such algorithms
are differentially private, we would like to note that making assumptions
about externally available information is usual (e.g. see invariants in~\cite{issues}).

A randomized query function $\kappa$ gives $\epsilon$-differential
privacy if, for all
data sets $D_{1}$, $D_{2}$ such that one can
be obtained from the other by modifying a single record,
and all $S\subset Range(\kappa)$, it holds 
\begin{equation}
\label{dp1}
\Pr(\kappa(D_{1})\in S)\le\exp(\epsilon)\times \Pr(\kappa(D_{2})\in S).
\end{equation}
In plain words, the presence or absence of any single record
is not noticeable (up to $\exp(\epsilon)$) when seeing the outcome of the query.
Hence, this outcome can be disclosed without
impairing the privacy of any of the potential respondents
whose records might be in the data set.
A usual mechanism 
to satisfy Inequality (\ref{dp1}) is to add noise to the 
true outcome of the query, in order to obtain an outcome of $\kappa$ 
that is a noise-added version of the true outcome. The 
smaller $\epsilon$, the more noise is needed to 
make queries on $D_1$ and $D_2$ indistinguishable up 
to $\exp(\epsilon)$.

In~\cite{Wang16,Wang14}, a connection between randomized response 
and differential privacy is established: 
randomized response is $\epsilon$-differentially private if 
\begin{equation}
\label{eqdp}
e^\epsilon \geq \max_{v=1,\ldots,r} \frac{\max_{u=1,\ldots,r} p_{uv}}{\min_{u=1,\ldots,r} p_{uv}}.
\end{equation}
The rationale is that the values in each column $v$ ($v \in \{1,\ldots,r\}$)
of matrix ${\bf P}$ correspond to the probabilities of the reported
value being $Y=v$, given that the true value is $X=u$ 
for $u \in \{1,\ldots,r\}$. 
Differential privacy requires that the 
maximum ratio between the probabilities in a column be bounded 
by $e^{\epsilon}$, so that the 
influence of the true value $X$ on the reported value $Y$ 
is limited. Thus, the reported value can be released
with limited disclosure of the true value.

\subsection{Frequency estimation error}
\label{sub:back_freq_error}

%Jordi3.
We want to minimize the error in the estimation of $\hat{\bf \pi}$. Following 
Equation (\ref{unbiased}),
%JOSEP4. Suprimit.
% we have $\hat{\bf \pi}  = ({\bf P}^T)^{-1} \hat{\bf \lambda}$,
%thus 
the error in $\hat{\bf \pi}$ comes from two sources: (i) the error in 
the estimation
of $\hat{\bf \lambda}$, and (ii) the propagation of that error when
computing the product
$({\bf P}^T)^{-1} \hat{\bf \lambda}$.

Following~\cite{agrawal2005}, the propagation error 
%JOSEP4. Poso lower-bounded
is lower-bounded by $P_{max}/P_{min}$,
where $P_{max}$ and $P_{min}$ are the maximum and minimum 
eigenvalues of ${\bf P}^T$.
Indeed,~\cite{agrawal2005} show that to minimize the 
propagation of the error the 
randomization matrix must have the form
\[{\bf P} = \left( \begin{array}{cccc} 
p_{u} & p_d & \cdots & p_d \\
p_d   & p_u & \ddots & \vdots \\
\vdots& \ddots &\ddots & p_d \\
p_d & \cdots & p_d &p_u
\end{array}\right),\]
%JOSEP4. Afegit
with $p_u \geq p_d$.
Provided we use a randomization matrix that minimizes the propagation of the
error, the error in $\hat{\bf \pi}$ is a function of the
error in $\hat{\bf \lambda}$. In the rest of this section, we bound the error in
the estimation of  $\hat{\bf \lambda}$.

%%JOSEP3. Reescrit.
%In RR the relative frequencies $\pi$ of the values of the true response $X$ 
%are unknown. 
%%JOSEP3. Suprimeixo,
%%Due to randomization,  
%%$\pi$ is transformed into $\lambda$ (frequencies of the randomized values),
%%which are also unknown. 
%Due to randomization, one observes a sample of the reported response $Y$.
%The sample proportions of the values of $Y$ are $\hat{\lambda}$,
%while the population proportions $\lambda$ stay unknown. 
%This is why $\pi$ is estimated based on $\hat{\lambda}$ using 
%Equation (\ref{unbiased}). 

We can view the sample of $Y$ 
as a draw from a multinomial distribution 
with $n$ trials (the data set size) and probabilities $\lambda$. Thus, it is
natural to derive the estimate  $\hat{\lambda}$ 
from the observed frequencies of $Y$.
%JOSEP3. Esborrat. Ja s'ha dit.
%The accuracy of $\hat{\lambda}$ determines the accuracy of 
%$\hat{\pi}$. Even if $\hat{\pi}$ is likely to be less accurate
%than $\hat{\lambda}$
%due to the randomization, in our analysis presented in 
%Section~\ref{sec:basic}, we use the  
%accuracy of $\hat{\lambda}$ as approximation of the accuracy of $\hat{\pi}$. 
Let us deal with the accuracy of $\hat{\lambda}$,
which in~\cite{Thompson87} 
is measured using confidence intervals as follows. 

\begin{defn}
	The absolute error of $\hat{\lambda}$ as an estimation of $\lambda$ is $e_{abs}$
	with confidence $\alpha$ if
	\[\Pr [\bigcap_{u=1,\ldots,r}(\hat{\lambda}_u-e_{abs} \le \lambda_u \le \hat{\lambda}_u + e_{abs})]\ge 1-\alpha.\]
\end{defn}

The absolute error can be determined, based on the frequencies and the data set size, as
\begin{equation}
\label{abserr}
 e_{abs} = \max_{u=1,\ldots,r} \sqrt{B\lambda_u (1-\lambda_u)/n},
\end{equation} 
where $B$ is the $\alpha/r$ upper percentile of the $\chi^2$ distribution with 1 degree of freedom.

The absolute error grows with $\sqrt{B}$, which in turn grows with
the number of categories $r$ as shown in Figure~\ref{fig:B}.
While the impact $r$ 
on the absolute error via $\sqrt{B}$ seems limited, 
the actual effect of increasing $r$ on 
Expression (\ref{abserr}) can be greater.
 The reason is that increasing the number of categories decreases
the absolute frequency of each category, 
which makes the {\em relative} error of $\hat{\lambda}$ more noticeable.
\begin{figure}
	\centering
	\includegraphics[width=0.8\linewidth]{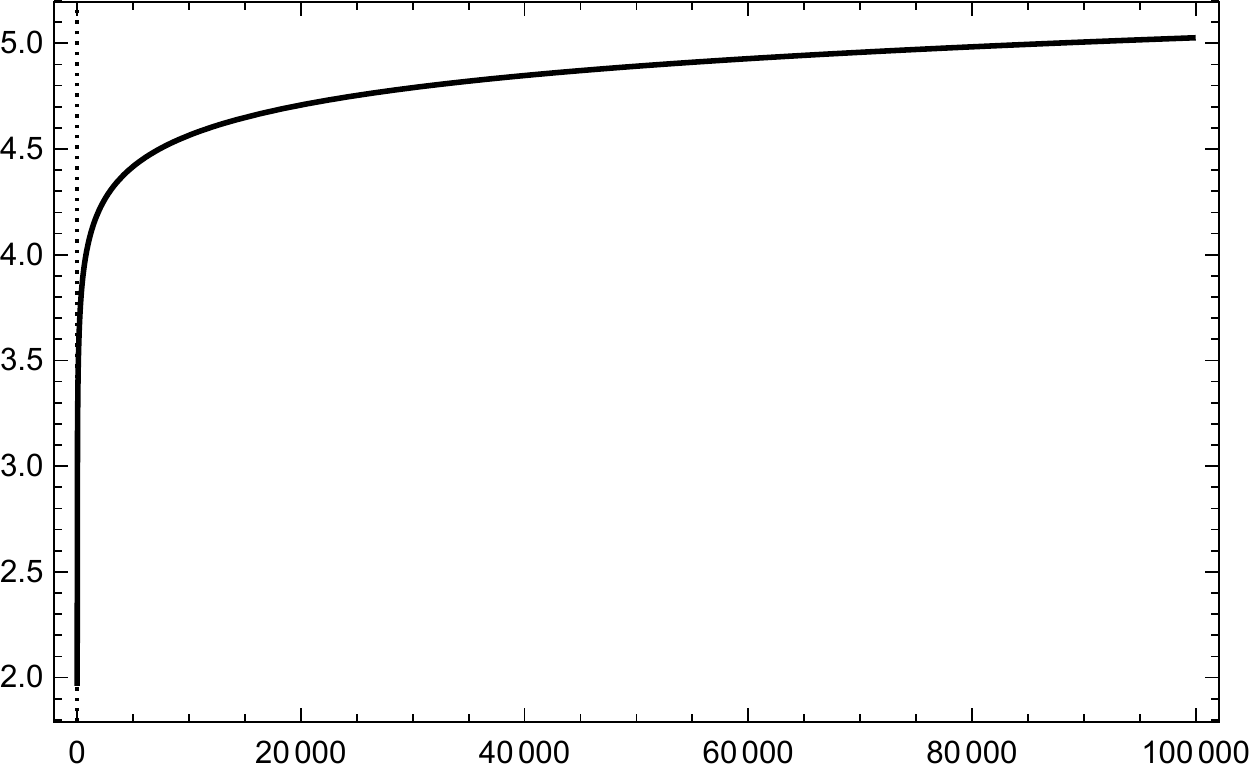}
	\caption{Evolution of the factor $\sqrt{B}$ ($y$-axis) of 
the absolute error of $\hat{\lambda}$ in terms of the number of 
		categories $r$ ($x$-axis) when $\alpha=0.05$}
	\label{fig:B}
\end{figure}

\begin{defn}
	The relative error of $\hat{\lambda}$ as an estimation of $\lambda$ is $e_{rel}$
	with confidence $\alpha$ if
%JOSEP3. IMPORTANT. Corregit el lloc de la titlla sobre lambda
	\[ \Pr [\bigcap_{u=1,\ldots,r}( (1-e_{rel})\hat{\lambda}_u \le \lambda_u \le (1+e_{rel})\hat{\lambda}_u]\ge 1-\alpha. \]
\end{defn}

Analogously to Expression (\ref{abserr}), we can 
determine the relative error based on the
 frequencies and the data set size:
%JOSEP3. Pas a equation perquè té etiqueta.
\begin{equation}
e_{rel} = \max_{u=1,\ldots,r} \sqrt{B\frac{1-\lambda_u}{\lambda_u}\frac{1}{n}},\label{eq:rel_err}
\end{equation} 
	where $B$ is the $\alpha/r$ upper percentile of the $\chi^2$ distribution with 1 degree of freedom.

\section{Basic RR protocols: RR-Independent and RR-Joint}
\label{sec:basic}

Assume $n$ parties $i=1,\ldots, n$ 
each holding one record ${\bf x}_i=(x^1_i,$ $\ldots,$ $x^m_i)$ 
that contains the values for $m$ attributes.
A data controller wants to perform exploratory analysis 
and/or machine learning tasks on the pooled data of the $n$
parties. However, no party wants to disclose her true record
even if she is ready to disclose a masked version of it.

In the above situation RR is a good option, as motivated
in the previous sections.
In this section we describe two basic 
RR methods to estimate the pooled true data of the parties, 
and we also highlight the limitations of such methods.

\subsection{RR-Independent}
\label{sub:basic_rr_ind}
This is the most naive solution. 
Each party $i$ separately deals with each attribute value $x_i^j$
for $j=1,\ldots,m$ via RR.
If the $j$-th attribute $A^j$ can take $r^j$ different
values, then an $r^j \times r^j$ probability matrix
${\bf P}^j$ (see Expression (\ref{probmatrix})) can 
be used for each party to report a 
randomized value $y_i^j$ for $A^j$ instead of her 
true value $x_i^j$. See Protocol~\ref{alg:RR-indep}.

%JOSEP3. canvio P^j per P_j
\begin{protocol}
	\caption{RR-Independent}
	\label{alg:RR-indep}
	\begin{enumerate}
		\item {\bf Randomization Protocol}
		\item Let $x_i^j$ be the value of party $i$ for attribute $j$.
		\item Let ${\bf P}^j$ be the randomization matrix for attribute $j$.
		
		\item Each party $i=1,\ldots,n$ applies RR with matrix ${\bf P}^j$
		to $x_i^j$, for $j=1,\ldots,m$, and publishes the result.
		\vspace{0.3cm}
		\item {\bf Attribute Distribution Estimation}
		\item Let $\hat{\lambda}^j$ be the experimental distribution of the randomized attribute $A^j$.
		\item The distribution of $A^j$ is estimated as $\hat{\pi}^j=(({\bf P}^j)^T)^{-1}\hat{\lambda^j}$.
		\vspace{0.3cm}
		\item {\bf Joint Distribution Estimation}
		\item Let ${\bf S}\subset A^1\times\ldots\times A^m$ be a subset of the data domain.
		\item The frequency of ${\bf S}$ is estimated as 
%JOSEP3. Arreglo notació.
		$\sum_{(x^1,\ldots,x^m)\in {\bf S}} \hat{\pi}^1(x^1)\times \ldots \times \hat{\pi}^m(x^m)$.
	\end{enumerate}
\end{protocol}

%JOSEP3. Corregida notació.
As mentioned in Section~\ref{sub:back_RR}, this would allow 
all parties to approximate the {\em marginal} empirical 
distribution ${\bf \pi}^j=(\pi^j_1,\ldots, \pi^j_{r^j})$ of 
each attribute $A^j$ as 
\[ \hat{\bf \pi}^j  = (({\bf P}^j)^T)^{-1} \hat{\bf \lambda}^j, \]
where $\hat{\bf \lambda}^j$ is the empirical distribution of attribute $A^j$ in 
the data set ${\bf Y}$ containing the randomized responses.

%JOSEP3. Corregida notació.
The problem is that estimating the marginal empirical distributions
of attributes does not yield in general an estimate of the joint
empirical distribution of the data set ${\bf X}$ formed by 
the true responses.
Only if the attributes in ${\bf X}$ are (nearly) independent
can their joint distribution be estimated from the 
marginal attribute frequencies. 
In this case, the frequency
of a set ${\bf S}\subset A^1 \times \ldots \times A^m$ can be estimated as
$\sum_{(x^1,\ldots,x^m)\in {\bf S}} \hat{\pi}^1(x^1)\times \ldots \times \hat{\pi}^m(x^m)$.
Clearly, the more dependent the attributes, the less accurate
is the previous estimate.

%Jordi202007. No poso el cost de la randomització perquè és petit respecte a la resta. En el pitjor dels casos 
% el cost de randomitzar cada atribut seria proporcional al nombre d'atributs però aprofitant les regularitats
% de la matriu es poden buscar mètodes més eficients. 
With respect to the computational cost, the fact that each attribute is dealt with
separately is positive because the randomization matrices remain small. The computational
complexity for each of the individuals that participate in the protocol is as follows:
\begin{itemize}
	\item Estimating the distribution of the true values of 
attribute $A^j$ as per Expression (\ref{unbiased}) 
amounts to computing the inverse of an $|A^j|$-dimensional 
	matrix followed by an $|A^j|$-dimensional matrix-vector product. 
%As far as the order is concerned we can overlook
%	the less costly the matrix-vector product, and conclude a cost of 
%	$\mathcal{O}(inverse~of~a~|A^j|-dimensional~matrix)$.
	The actual cost is dominated by the matrix inversion. 
%JOSEP202007. Afegida citació Strassen.
If using Strassen's algorithm (the best-performing
	practical algorithm, ~\cite{strassen}), inversion takes $\mathcal{O}(|A^j|^{2.807})$. In the particular case of the randomization matrices described in 
	Section~\ref{sub:back_freq_error}, their regularity makes it possible to easily compute their inverses with a
	cost $\mathcal{O}(|A^j|^2)$.
	\item The cost of estimating the joint frequency of one combination of attribute values is $\mathcal{O}(m)$. 
\end{itemize}

\subsection{RR-Joint}
\label{sub:RR-Joint}

To estimate the frequency of an arbitrary set 
%JOSEP4. S amb negreta.
${\bf S}\subset A^1 \times \ldots \times A^m$
without requiring attribute independence, 
we need to directly estimate the joint distribution.
To do this via RR, 
each party must report her randomized response 
for the value of $A^1 \times \ldots \times A^m$. 
%JOSEP3. Afegit.
After this, the frequency
%JOSEP4. S negreta
of a set ${\bf S}\subset A^1 \times \ldots \times A^m$ can be estimated as
$\sum_{{\bf x}\in {\bf S}} ({\bf P}^T)^{-1} \hat{\bf \lambda}({\bf x})$,
where $\hat{\bf \lambda}(\bf x)$ is the empirical distribution of {\bf x} in ${\bf Y}$.
See Protocol~\ref{alg:RR-Joint}.

Two observations are in order. First, thanks to RR,
all parties preserve the confidentiality of their true inputs
during Protocol~\ref{alg:RR-Joint}. 
Second, once the estimate of the empirical
joint distribution of ${\bf X}$ is published,  
any parties can perform statistical computations
on it; they can even create a synthetic 
data set by repeating each combination of 
$A^1 \times A^2 \times \ldots \times A^m$ as many times
as dictated by its frequency in the joint distribution.

\begin{protocol}
	\caption{RR-Joint\label{alg:RR-Joint}}
	\begin{enumerate}
		\item {\bf Randomization Protocol}
		\item Let ${\bf x}_i$ be the record of party $i$.
		\item Let ${\bf P}$ be the randomization matrix of $A^1 \times \ldots \times A^m$.
		\item Each party $i=1,\ldots,n$ runs RR with matrix ${\bf P}$
		on ${\bf x}_i$ and publishes the result.
		\vspace{0.3cm}
		\item {\bf Joint Distribution Estimation}
		\item Let ${\bf S}\subset A^1\times\ldots\times A^m$ be a subset of the data domain.
		\item The frequency of ${\bf S}$ is estimated as 
		$\sum_{{\bf x}\in {\bf S}} ({\bf P}^T)^{-1} \hat{\lambda}({\bf x})$.
	\end{enumerate}
\end{protocol}

Unfortunately, direct estimation of the joint distribution is not 
without limitations. As the number of attributes grows linearly, the number
of categories of the Cartesian product grows exponentially.
This causes the computational cost to increase and the 
 accuracy of the frequency estimates to decrease in ways
that are not acceptable.

%Jordi202007. No posso el cost computacional en aquest cas perquè ja es deia això.
The high computational cost comes from having a 
vector of frequencies
of exponential size $\Pi_{j=1,\ldots,m}|A^j|$ and a 
(huge) randomization matrix with 
 $\Pi_{j=1,\ldots,m}|A^j|$ rows and the same number of columns. 
According to Expression (\ref{unbiased}), to estimate 
$\hat{\pi}$ we need to multiply the inverse of the 
transpose of the randomization matrix times 
$\hat{\lambda}$. Even if we assume that 
the (computationally costly) inverse matrix is available so that we only need to perform the matrix multiplication, 
the cost remains exponential in the number of attributes.

Regarding the accuracy of the frequency estimates, 
direct estimation of the joint distribution only works well if the number of 
parties $n$ is much larger than the number of possible
values of the above Cartesian product, that is, when
\begin{equation}
\label{nbound}
n \gg |A^1| \times |A^2| \times \ldots \times |A^m|.
\end{equation}
The necessity of Bound (\ref{nbound}) becomes obvious 
when we analyze the error as per
Expression (\ref{eq:rel_err}) for $n = |A^1| \times |A^2| 
\times \ldots \times |A^m| = r$.
Even if frequencies $\lambda_u$  
were evenly distributed and equal 
to $1/r$ ---which would minimize 
their relative estimation error---,
 Expression (\ref{eq:rel_err}) would yield approximately $\sqrt{B}$. 
By looking at Figure~\ref{fig:B}, 
we observe that $\sqrt{B}$ is too big (above $200\%$) to be acceptable as 
a relative error.

\subsection{Accuracy analysis}
\label{sub:accuracy}

Let us compare the relative error achieved 
by Protocols RR-Inde\-pend\-ent and RR-Joint. As above,
we will perform the analysis in the best case, that is, 
when frequencies are evenly distributed.

Thus, for R-Independent we take 
%JOSEP3. IMPORTANT. Corregit.
$\lambda^j_u = 1/|A^j|$ for 
all $j=1,\ldots,m$ and all $u=1,\ldots,|A^j|$.
According to Expression (\ref{eq:rel_err}), 
the relative error of attribute frequencies
in RR-Independent is
%JOSEP3. IMPORTANT. Em sembla que és dividit per n i no per n^2
%JOSEP3. IMPORTANT. Canviat B per B^j
\[e_{rel} = \max_{j=1,\ldots,m} \sqrt{B^j \frac{|A^j|-1}{n}},
\]
where $B^j$ is the $(\alpha/|A^j|)$ upper percentile of a $\chi^2$ distribution with one degree of freedom.

Similarly, the relative error of the estimated frequencies in RR-Joint is
\[e_{rel} = \sqrt{B \frac{\Pi_{j=1,\ldots,m}|A^j|-1}{n}},
\]
where $B$ is the $(\alpha/\Pi_{j=1,\ldots,m}|A^j|)$ upper percentile of a $\chi^2$ distribution with one degree of freedom.

Notice that the relative error grows as the square root of the number of categories, which is exponential in 
the number of attributes. Thus, RR-Joint is likely to have poor accuracy estimates already for a small number of
attributes. As described in Section~\ref{sub:RR-Joint}, this can only be mitigated by a large data set size $n$,
which becomes unrealistic already for a moderate number of attributes.

\section{RR-Clusters}
\label{sec:RR-clusters}

Neither Protocol RR-Independent nor Protocol RR-Joint are satisfactory,
the former due to the independence requirement and the latter
due to the combinatorial explosion of the number of categories. 
In this section,
we propose RR-Clusters, a protocol that 
 strives to use as little as possible 
the independence assumption, while keeping a
reasonable computational cost and estimation accuracy.

The RR-Clusters protocol splits attributes into clusters according to their 
mutual dependence and 
%JOSEP3. Dic que es fa RR-Joint.
{\em performs RR-Joint independently on each of the 
attribute clusters}.
 
 Running RR-Joint separately on each attribute cluster implies 
neglecting the possible dependences between 
attributes in different clusters. Thus, 
we need to determine clusters in such a way
that no significant dependence exists between any
 two attributes in different clusters. 
%to estimate 
%the joint distribution across 
%clusters, we need to  
%assume that attributes {\em in different clusters} are pairwise independent.
%To minimize the error introduced by this assumption, 
%we need 

Additionally, to keep the computational cost and the 
estimation error within reasonable bounds, we want
the cardinality of the Cartesian product of attributes 
within each cluster to 
be small compared with the number
of records of the data set.
Thus, the clustering algorithm should
try to place in the same cluster
only those attributes that have a strong mutual dependence.

More specifically, we want a set of clusters $C_1,\ldots,C_l$, 
for some $l$, such that:
\begin{itemize}
	\item $\bigcup_{k=1,\ldots,l}C_k = \{A^1,\ldots,A^m\}$ 
and $C_i \cap C_j =\emptyset$ for $i\ne j$;
	\item It holds that $n \gg \max_{k=1,\ldots,l}\Pi_{j\in C_k} |A^j|$;
	\item The dependence between attributes in different clusters is
as low as possible.
\end{itemize}  

Since attributes are to be clustered based on their dependences,
we start by describing the clustering algorithm assuming that 
those dependences are available.
Computing the dependences between attributes 
would be easy if one could
resort to a trusted party holding the entire data set. Lacking a trusted party, we need methods 
to assess attribute dependence without requiring parties to disclose their 
data. We describe such methods 
 in Sections~\ref{sub:rr_ind},
\ref{sub:sec_sum} and~\ref{sub:rr_pair}; they differ
in their accuracy and in the disclosure risk for 
the parties' true attribute values.

The clustering algorithm is formalized in Algorithm~\ref{alg:clusters}.
It starts with single-attribute
clusters. It then loops through the list of cluster 
pairs in descending order
of dependence, and merges two clusters  
if the number of combinations
of attribute values in the merged cluster remains below a given threshold. 
Additionally, to avoid
clustering attributes that are not really dependent, 
the algorithm uses a threshold on the dependence measure below
which clusters are not merged.
%JOSEP3. Afegit.
The dependence between two clusters of attributes is defined
as the maximum dependence between pairs of attributes such
that one attribute is in one cluster and the other attribute
is in the other cluster. 

\begin{algorithm}
%JOSEP3. Algorisme bastant reescrit.
	\caption{Clustering of attributes based on their dependence}
% $T_v$: upper threshold on the number of combinations of attribute
%values in a cluster. $T_d$: lower threshold on the dependence
%between attributes in a cluster.}
	\label{alg:clusters}
\begin{enumerate}
	 \item Let $T_v$ be the maximum number of combinations of 
attribute values allowed in a cluster
	 \item Let $T_d$ be the minimum dependence required between
two clusters for them to be merged 
\item Let $Clusters=\{\{A^1\},\ldots,\{A^m\}\}$ 
	 \item Let $DependenceList$ be the list of dependences 
between cluster pairs
	 \vspace{0.3cm}
	 
	 \item Sort $DependenceList$ in descending order
	\item Let $dep$ be the first element of $DependenceList$
	 \item While $dep \geq T_d$ do
	 	\item \hspace{0.3cm} Let $(C_1, C_2)$ be the cluster 
pair whose dependence is $dep$
	 	\item \hspace{0.3cm} If $\Pi_{A\in C_1 \cup C_2} |A| \le T_v$ then
	 	\item \hspace{0.6cm}    Remove $C_1$ and $C_2$ from $Clusters$
	 	\item \hspace{0.6cm}    Add $C_1 \cup C_2$ to $Clusters$
  \item \hspace{0.6cm}    Recompute $DependenceList$ for $Clusters$
\item \hspace{0.6cm}    Sort $DependenceList$ in descending order
\item \hspace{0.6cm} Let $dep$ be the first element of  $DependenceList$ 
 \item \hspace{0.3cm} Else 
\item \hspace{0.6cm} Let $dep$ be the next element of  $DependenceList$			\item \hspace{0.3cm} End if
	 \item End while
\end{enumerate}
	
\end{algorithm}

The specific measure of dependence to be used must take
into account the type of the attributes. We will use
the following dependence metrics adopted in~\cite{psd2018}.
If $A^i$ and $A^j$ are ordinal, we can take as a measure
of dependence 
\begin{equation}
\label{indnum}
|r_{ij}|,
\end{equation}
where $r_{ij}$ is Pearson's
correlation coefficient between $A^i$ and $A^j$. 
Expression (\ref{indnum}) can also be used for continuous
numerical attributes, but to be accommodated by RR these 
need to be discretized into ordinal attributes
(for example by
rounding or by replacing values with intervals).

If one of $A^i$ and $A^j$ is nominal (without 
an order relationship between its possible values) and the other
is nominal or ordinal, we can take as a measure of independence
\begin{equation}
\label{indcat}
V_{ij},
\end{equation}
where $V_{ij}$ is Cram\'er's V statistic~\cite{Cram46}, that gives 
a value between 0 and 1, with 0 meaning complete independence
between $A^i$ and $A^j$ and 1 meaning complete dependence.
Cram\'er's $V_{ij}$ is computed as 
\[ V_{ij} = \sqrt{\frac{\chi^2_{ij}/n}{\min(r^i -1, r^j-1)}}, \]
where $r^i$ is the number of categories of $A^i$, 
$r^j$ is the number of categories of $A^j$, $n$ 
is the total number of parties/records and 
$\chi^2_{ij}$ is the chi-squared independence statistic defined as
\[\chi^2_{ij} = \sum_{a=1}^{r^i} \sum_{b=1}^{r^j} \frac{(o^{ij}_{ab} - e^{ij}_{ab})^2}{f^j_{ab}},\]
with $o^{ij}_{ab}$ the observed 
frequency of the combination $(A^i=a, A^j=b)$ and $e^{ij}_{ab}$ 
the expected frequency of that combination under the independence
assumption for $A^i$ and $A^j$. This expected frequency is computed as
\[ e^{ij}_{ab} = \frac{n^i_a n^j_b}{n},\]
where $n^i_a$ and $n^j_b$ are, respectively, 
the number of parties who have reported $A^i=a$ and $A^j=b$.

Finally, if one of $A^i, A^j$ is nominal and the other 
is numerical, the latter must be discretized.
After that,
the contingency table between $A^i$ and $A^j$ can be 
constructed, and the measure of dependence given
by Expression (\ref{indcat}) can be computed.

Note that Expressions (\ref{indnum}) and (\ref{indcat}) are
bounded in $[0,1]$, and hence 
the outputs of both expressions are comparable when trying 
to cluster the attributes.

%Jordi202007. Cost computacional
%JOSEP202007. Retocat i canviat de lloc.
With respect to the computational cost, the fact that the number of combinations of attribute values within
each cluster is relatively small (because 
each cluster contains only a subset of attributes) 
leads to a set of randomization matrices 
of relatively small size; this is positive to
keep the computational cost reasonable. The computational cost for each of the individuals that participate 
in the protocol is as follows:
\begin{itemize}
	\item Estimating the joint distribution of 
the true values of attributes in cluster $C_k$ as per 
Expression (\ref{unbiased}) 
amounts to computing the inverse of a $\prod_{A\in C_k}|A|$-dimensional 
	matrix followed by a $\prod_{A\in C_k}|A|$-dimensional matrix-vector product. As far as the order is concerned we can overlook
	the less costly matrix-vector product. 
	As to the cost of computing the inverse matrix, with Strassen's
algorithm it is $\mathcal{O}(\prod_{A\in C_k}|A|^{2.807})$. 
	\item The cost of estimating the joint frequency of one combination of attribute values is $\mathcal{O}(l)$, where $l$ is the number of clusters. 
\end{itemize}

%Jordi3.
%JOSEP4. Lleugerament canviat.
Since ${\bf X}$ is formed by the true responses of the parties
and these do not disclose them, 
no single party can compute the dependences between
the attributes in ${\bf X}$. In the following subsections, we describe 
several methods to compute such dependences based on partial 
and/or inaccurate information submitted
by the individual parties. Each method has a different level of accuracy as well as a different 
impact on privacy.
The overall privacy risk is the aggregation of the risk that results 
from 
the computation of the dependences between attributes and the 
risk that results
from applying RR to the resulting clusters of attributes. 
In differential privacy terms,
%Jordi2020. Referencia afegida
the sequential composition property applies~\cite{composition}: 
if the computation of the dependences
between attributes is $\epsilon_1$-differentially private and 
the RR data release 
is $\epsilon_2$-differentially private, overall one has 
 $\epsilon_1+\epsilon_2$-differential privacy.

%Jordi2020. Un revisor es queixava de que aquestes seccions feien suposicions que no
%s'havien introduit abans. Ho intento explicar.
%Each of the following subsections describes a different approach to compute the 
%dependence between attributes. 
In Subsection~\ref{sub:rr_ind}, we describe an efficient 
method for computing the dependences 
between the attributes in ${\bf X}$  when the dependence measure is the covariance. Subsections~\ref{sub:sec_sum} and \ref{sub:rr_pair} deal with arbitrary dependence measures and differ in the adversary model. Taking advantage of a common distinction between confidential and quasi-identifying attributes, Subsection~\ref{sub:sec_sum} computes the dependence measure over each pairwise distribution of attributes. Subsection~\ref{sub:rr_pair} avoids such a distinction and leverages 
randomized response to compute the dependences between attributes.

\subsection{Randomized response on each attribute}
\label{sub:rr_ind}

%Jordi3.
In this section, we approximate the dependences between pairs of attributes
based on a data set ${\bf Y}$ obtained by independently 
randomizing each attribute in 
${\bf X}$ (see Section~\ref{sub:basic_rr_ind}).
%Since ${\bf X}$ is not available, we cannot compute the dependences between
%attributes in it. However, we can approximate them based on ${\bf Y}$.
Dependences between attributes in ${\bf Y}$ are likely
to be attenuated versions of those in ${\bf X}$,
but as long as the former dependences preserve 
the ranking of the latter, attribute clustering
based on the former should be fine.
In other words, if in ${\bf X}$ the dependence 
between $A^1$ and $A^2$ is stronger than 
the dependence between $A^3$ and $A^4$,
our requirement is that the same relation hold in ${\bf Y}$.

To analyze the effect of RR on the dependence between
attributes, we view each attribute of ${\bf X}$ as a random variable
whose distribution is the empirical distribution of the attribute
values.
Note that this view is not entirely accurate, because 
it is unlikely that the same set of values is obtained if  
the attribute's random variable is sampled.
However, this is a useful approximation that 
significantly simplifies the analysis.

Proposition~\ref{prop:cov} examines 
the effect on the covariance between attributes when RR
is independently run on each attribute.
This proposition is later used in Corollary~\ref{cor:cov} to show that 
when RR is appropriately run, 
the relative strength of the covariances between pairs of attributes
is not altered.

%JOSEP. Canvio notació per uniformitzar amb el que precedeix.
\begin{prop}\label{prop:cov}
	Let $X^a$ and $X^b$ be two finite random variables
	that are 
independently randomized into $Y^a$ and $Y^b$, respectively, 
	as follows:
	\begin{itemize}
		\item With probability $p_a$ (resp. $p_b$),  $Y^a:=X^a$ (resp. 
		$Y^b:=X^b$).
		\item With probability $1-p_a$ (resp. $1-p_b$), $Y^a:=U^a$ 
		(resp. $Y^b:=X^b$), where $U^a$ (resp. $U^b$) is a random variable that is uniformly distributed on the support of $X^a$ (resp. $X^b$). 
	\end{itemize}
	Then the covariance of $Y^a$ and $Y^b$ is
	$Cov(Y^a,Y^b)=p_a \times p_b \times Cov(X^a,X^b)$.
\end{prop}

\begin{proof}
	The covariance $Cov(Y^a,Y^b)$ can be expressed as 
	\begin{equation}
	\label{covar}
	Cov(Y^a,Y^b)=E(Y^a Y^b)-E(Y^a)E(Y^b).
	\end{equation}
	
	Let us start by computing $E(Y^a)E(Y^b)$. 
	The expected values of $Y^a$ and $Y^b$ are
	$ E(Y^a)=p_a  E(X^a) + (1-p_a) E(U^a)$
	and 
	$E(Y^b)=p_b  E(X^b) + (1-p_b) E(U^b)$.
	Their product is
	\begin{align}
	E(Y^a)E(Y^b) &= p_a p_bE(X^a)E(X^b) \nonumber\\
	&+ p_a(1-p_b)E(X^a)E(U^b) \nonumber\\
	&+ (1-p_a)p_bE(U^a)E(X^b) \nonumber\\
& +(1-p_a)(1-p_b)E(U^a)E(U^b).
	\label{primera}
	\end{align}
	
	Let us now compute $E(Y^a Y^b)$. The random variable $Y^a Y^b$ can be
	expressed as:
	\[
	Y^a Y^b=\begin{cases}
	X^a X^b & \text{with probability }p_{a}p_{b};\\
	X^a U^b & \text{with probability }p_{a}(1-p_{b});\\
	U^a X^b & \text{with probability }(1-p_{a})p_{b};\\
	U^a U^b & \text{otherwise.}
	\end{cases}
	\]

%JOSEP3. Corregides un parell de probabilitats als dos desenvolupaments
%que segueixen.	
	Hence, $E(Y^a Y^b) = p_a p_b E(X^a X^b) + p_a(1-p_b)E(X^a U^b) +
	(1-p_a)p_bE(U^a X^b) + (1-p_a)(1-p_b)E(U^a U^b)$. 
Since $U^a$ and $U^b$ are
	independent from each other and also are independent from $X^a$ 
and $X^b$, we can
	write the expected values as products:
	
	\begin{align}
	E(Y^a Y^b) &= p_a p_b E(X^a X^b) + p_a(1-p_b)E(X^a)E(U^b) \nonumber\\
	& +(1-p_a)p_bE(U^a)E(X^b) \nonumber\\
        &  + (1-p_a)(1-p_b)E(U^a)E(U^b).
	\label{segona}	
	\end{align}
	
	By plugging Expressions (\ref{primera}) and (\ref{segona})
	into Expression (\ref{covar}) and simplifying, we can express the 
	covariance as:
	\begin{align}
	Cov(Y^a,Y^b) &=p_a p_b (E(X^a X^b)-E(X^a)E(X^b))\nonumber\\
	&=p_a p_b Cov(X^a,X^b).\nonumber
	\end{align}
\end{proof}

\begin{corollary}
	\label{cor:cov}
	%JOSEP2. Reescrit.
	Let ${\bf X}$ be a data set with attributes $X^1,\ldots,X^m$,
	that are randomized into $Y^1,\ldots,Y^m$ as follows: 
	\[
	Y^j=\begin{cases}
	X^j & \text{with probability }p,\\
	U^j & \text{otherwise,}
	\end{cases}
	\]
	where $U^j$ is uniformly distributed over the support of $X^j$ 
	and $p\in (0,1]$.
	Then the randomization does not alter the relative strength of the covariance between
	attributes, that is, if $Cov(X^j,X^k)>Cov(X^l,X^m)$ then  $Cov(Y^j,Y^k)>Cov(Y^l,Y^m)$.
\end{corollary}
\begin{proof}
	The proof immediately follows from Proposition~\ref{prop:cov}:
	\begin{align}
	Cov(Y^j,Y^k)&= p^2 Cov(X^j,X^k)>p^2Cov(X^l,X^m) \nonumber\\
	&=Cov(Y^l,Y^m). \nonumber
	\end{align}
\end{proof}

According to Proposition~\ref{prop:cov}, RR attenuates the covariance
between attributes, but Corollary~\ref{cor:cov} shows that 
it preserves the relative strength of the covariances between
pairs of attributes. 
Since when clustering attributes 
with Algorithm~\ref{alg:clusters} we are  interested in the relative strength of the 
dependence between attributes, it makes sense to run Algorithm~\ref{alg:clusters} using dependences between randomized attributes. 
%JOSEP3. IMPORTANT. Afegeixo tot això per clarificar.
Thus, attribute clusters can be obtained as follows:
\begin{enumerate}
\item Every party publishes the value of each of her attributes using
RR with the randomization of Corollary~\ref{cor:cov}.
\item Dependences are computed on the randomized attributes.
\item Algorithm~\ref{alg:clusters} computes attribute
clusters based on the dependences
between randomized attributes. 
\end{enumerate}

%JOSEP202007. Retocat.
%Jordi2020. Afegit el cost
%Doing as just explained boils down to performing randomized 
%response over each of the $m$ attributes.
%JOSEPNOV20.
As to the communication cost of this method, 
at least one individual 
must gather the entire randomized data set in 
order to compute attribute dependences and thereby 
the attribute clustering, which is then
shared with the rest of individuals. Since the size
of the randomized data set is $O(nm)$, the communication
cost for the individual(s) computing the clustering is $O(nm)$. For 
the rest of individuals, the communication 
cost is only $O(m)$: they have to send 
their randomized record, then receive
the attribute clustering and finally return their
randomized record according to the received clustering.
%Jordi202007. Cost computacional de l'algorisme. No compto el cost de recollir 
% la informació de cada individu, ja que això seria cost de comunicació.
%JOSEPNOV20. Parlo de les dependències i dels clusters.
Regarding the computational cost, the computations are fairly
simple for most individuals: they only need to randomize their records, 
which takes cost $O(m)$. For the 
individual(s) computing the dependences and the 
 clustering of the $m$ attributes there is an additional
cost that depends on the dependence measure and the clustering
algorithm used.

%Jordi202007. No tinc en compte el cost de calcular la covariança entre atributs perquè
% als altres protocols no es parla de covariança.

%Jordi3. Mogut aquí
In the above proposition and corollary, 
we have focused on covariance to measure dependence.
Intuition tells us that the effect of randomization on 
other dependence measures, such as the ones mentioned
in Section~\ref{sec:RR-clusters} for use with Algorithm~\ref{alg:clusters},
can be expected to be similar: attenuation of dependence
but preservation of its relative strength. However, this need
%JOSEP4. Afegida menció secció següent.
not always be the case ---see Section~\ref{sub:sec_sum}
for a method that is agnostic of the dependence measure.

Once attribute clusters have been determined, 
parties can use RR-Joint within each cluster, which yields
the final randomized data set. 

%Jordi3. Nou
Due to the fact that RR yields differential privacy (see Section~\ref{sec3}),
this method to compute the dependence between attributes yields differential
privacy.
%However, the proposed method is not without 
%issues:
%\begin{itemize}
%	\item In the above proposition and corollary, 
%	we have focused on covariance to measure dependence.
%	Intuition tells us that the effect of randomization on 
%	other dependence measures,  such as the ones mentioned
%	in Section~\ref{sec:RR-clusters} for use with Algorithm~\ref{alg:clusters},
%	can be expected to be similar: attenuation of dependence
%	but preservation of its relative strength. However, this need
%	not always be the case.
%	\item The use of RR on each attribute to estimate the dependence 
%	between pairs of attributes has the drawback that 
%	each party is forced to publish her randomized data twice:
%	once to estimate the dependence  
%	and once for the final
%	data release. This is problematic if 
%	an intruder can link both RR releases
%	to improve his knowledge about the true values of a given party.
%	To evaluate the privacy loss incurred if the two releases
%	are linked, we 
%	can use the equivalence with differential privacy 
%	(Section~\ref{sec3}) and the sequential composition 
%	property of DP: if one RR release
%	gives privacy $\epsilon_1$ and the other RR release $\epsilon_2$,
%	the privacy level when combining both of them is $\epsilon_1+\epsilon_2$.
%	If this resulting privacy is too low, then the randomization
%	level must be increased to reduce the probability 
%of reporting true values, thereby reducing $\epsilon_1$ and $\epsilon_2$.
%\end{itemize}

\subsection{Exact bivariate distribution via secure sum}
\label{sub:sec_sum}

%Jordi3.
The method described in the previous section does not generalize to an 
arbitrary measure of dependence between attributes. 
%To avoid the issues raised in Section~\ref{sub:rr_ind},
%we can think of an alternative that exploits the following two facts:
In this section, we propose an alternative method that exploits the following two facts:

\begin{itemize}
	\item We can remove identifier attributes from ${\bf X}$ 
	without affecting its statistical utility;
	\item Bivariate distributions are enough to compute the dependence 
	between pairs of attributes.
\end{itemize}

Essentially, the alternative procedure consists of 
each party releasing her true values 
for each pair of attributes, which yields bivariate
distributions wherefrom attribute dependences are obtained. 
Note that, since 
parties release unmasked data, differential privacy
does not apply. In spite of that, the risk of disclosure is low.
Since ${\bf X}$ does not include identifiers, 
we are in one of the following situations:
\begin{itemize}
	\item If none of the attributes of a given pair is confidential, then there is no risk of disclosure.
	\item If there is one confidential attribute (or two) in the pair, 
	then intruders cannot re-identify the party to whom the 
	record corresponds because there is (at most) 
	one quasi-identifier attribute in the pair. 
Recall that a single quasi-identifier is
	not enough to re-identify a record (otherwise it would be 
	an identifier).
\end{itemize}

The fact that a pair of attribute values is not re-identifying 
is of little help if the sender of the pair can be traced. 
Thus, we need the following:
\begin{itemize}
	\item {\em Anonymous communication.} 
	The communication channel should be anonymous. Otherwise, 
	the identity of the party to whom a pair corresponds can be 
	established by an intruder.
	\item {\em Unlinkability of communications originated by a party.} 
	If an intruder can link several pairs of
	values originated by a party, 
	he can acquire the values of several quasi-identifiers for that party, 
	which may lead to her re-identification.
	Further, if confidential attributes are also among the 
	acquired attributes, confidential information on the re-identified
	party has been disclosed.  
\end{itemize} 

The above properties can be attained by using a secure sum protocol. 
To illustrate how this can work to obtain
the empirical distribution of two categorical attributes,
we give a simple secure sum protocol that instantiates 
the general framework of~\cite{BGW}.
Let $(a,a')$ be 
a possible combination of values of attributes $A$ and $A'$. 
To compute the absolute frequency of $(a,a')$,
$n$ parties proceed as follows:

\begin{enumerate}
	\item Each party $i$ chooses a set of $n$ 
	random numbers $r_{ij}$ such that 
%JOSEP3. Canvio la notació de múltiple per la de mòdul, pq el punt 
%no es veu.
	$\sum_{j=1,\ldots,n}r_{ij} \pmod {n+1}=0$, that is,
so that the sum of the chosen numbers is a multiple of $n+1$;
	\item Each party $i$ sends $r_{ij}$ to party $j$, for each $j$;
	\item Each party $j$ collects $r_{1j},\ldots,r_{nj}$, computes 
	$\sum_{i=1,\ldots,n}r_{ij}$, and broadcasts:
	\begin{itemize}
		\item $r_j=\sum_{i=1,\ldots,n}r_{ij} + 1$ if 
%JOSEP3. IMPORTANT. Aquí deia "for party i" i ha de ser "for party j", crec.
		attributes $A$ and $A'$ take the values $a$ and $a'$, 
respectively, for party $j$;
		\item $r_j=\sum_{i=1,\ldots,n}r_{ij}$, otherwise;
	\end{itemize} 
	\item Each party $i$ collects $r_j$ for $j=1,\ldots,n$ and computes the absolute frequency of $(a,a')$ as
	$\sum_{j=1,\ldots,n}r_j \pmod{n+1}$.
\end{enumerate}

%Jordi2020.
%JOSEP2020. Revisat.
Computing the frequency of every combination of values for each
pair of attributes
 increases the communication cost with respect to the method 
of Subsection~\ref{sub:rr_ind}. Taking into account that the cost of the secure-sum protocol is proportional to the number of individuals $n$ and that it 
%Jordi202007. Hi havia un error. S'ha de fer per cada possible valor de cada parell d'atributs.
must be run for each possible value of each pair of attributes, the communication cost is 
$O(\sum_{1\le i<j\le m}(|A^i||A^j|) \times n)$.
%Jordi202007.
The computations needed to run the method in this subsection 
are fairly simple. For this reason, the overall cost 
is dominated by the communication cost.

\subsection{Randomized response on each pair of attributes}
\label{sub:rr_pair}

The procedure described in Section~\ref{sub:sec_sum} makes sure that 
only bivariate distributions are made available. 
Since a distribution does not report information 
on any specific party, its publication should in principle be safe. However,
small frequencies are problematic as they may enable linking 
two or more pairs of values and this may lead to re-identification
and disclosure.
For example, if upon seeing the bivariate distribution
of attribute $A^1$ and some other attribute $A^2$, an intruder 
learns there is a single party with $A^1=x$,
then he can link all pairs $(A^1=x,A^j)$ for $j \neq 2$ 
that have nonzero frequency
and thereby rebuild the party's entire
record, which may lead to the party's re-identification.

To prevent the above from happening, we can use RR on 
each pair of attributes, which
%Jordi3. 
%JOSEP4. Canviada referència a secció. 
renders 
the computation differentially private (see Section~\ref{sec3}). 
The procedure is
as follows:

\begin{enumerate}
	\item Let $A^1,\ldots,A^m$ be the attributes of data set ${\bf X}$;
	\item Let ${\bf P}^{ij}$ be a  randomization matrix for the pair of attributes $(A^i,A^j)$, for $1 \leq i<j \leq m$;
	%JOSEP2. Afegit el for
	\item For every pair $(A^i,A^j)$ of attributes:
	\begin{enumerate}
		\item Each party uses ${\bf P}^{ij}$ to mask her value for 
		the pair via RR;
%JOSEP3. IMPORTANT. Abans dèiem que el party publicava el parell RR, 
%però no n'hi ha cap necessitat. Més val que no el reveli.
% and releases the masked pair;
		\item The $n$ parties engage in the secure sum protocol described in Section~\ref{sub:sec_sum} to compute the distribution of the masked 
attribute pair; 
		\item Each party uses Equation (\ref{unbiased}) 
		to estimate the empirical distribution of the {\em unmasked} pair of attributes
		$(A^i,A^j)$ in ${\bf X}$.
	\end{enumerate}
\end{enumerate} 

Once the above procedure is complete, the dependence
between any two attributes $A^i$ and $A^j$ can
be assessed by all parties based on the 
estimated empirical distribution of $(A^i,A^j)$. 

Notice that, in spite of using RR, the above procedure
still resorts to 
the secure sum procedure: the purpose is to make 
each masked pair unlinkable to the party that originated it 
and also to make masked pairs corresponding to the same 
party unlinkable between them. 
Indeed, in the above procedure RR
is computed $m-1$ times on each attribute 
(once for each other attribute in ${\bf X}$); 
if an intruder was able to link a party's $m-1$
masked responses, the risk of disclosing the party's true 
responses and identity 
would increase significantly.

If the randomization matrix is chosen adequately,
this method to compute the dependence between attributes
can achieve differential privacy, as recalled 
in Section~\ref{sec3}. Strictly
%JOSEP3. IMPORTANT. Vols dir que ara cada atribut es publica m-1 cops?
%Es fa implícitament dins de la suma segura, però vaja...
speaking, since each attribute
is randomized and released $m-1$ times in the secure sum, 
sequential composition tells that the
overall level of differential privacy is the sum of the 
levels of each release.
However, the secure sum makes releases unlinkable, 
so an intruder cannot take advantage of the 
multiple releases to increase his knowledge about 
the value of an attribute. In this situation, 
we can waive sequential composition and
take the overall level of differential privacy to be same
as if each attribute were released only once:
%Jordi3.
the unlinkability property 
 closely matches the requirements of parallel composition.

%Jordi2020.
%JOSEP2020. Revisat.
%Jordi202007. Canviat
The cost of the method in this subsection 
is easily computed by noticing that it is simply the
method described in Subsection~\ref{sub:sec_sum} with an additional randomization step.
As the cost of the randomization step is not significant with respect to the cost of
the secure sum, the cost of the method here equals that of Subsection~\ref{sub:sec_sum}:
$O(\sum_{1\le i<j\le m}(|A^i||A^j|) \times n)$.

\section{RR-Adjustment}
\label{sec:Adjustment}

To estimate the joint distribution in RR-Independent, we needed to 
assume that attributes were independent. We then proposed
 RR-Clusters to partially circumvent that need. We write
``partially'' because RR-Clusters still needs 
to  assume that attributes in different clusters are independent. 
The purpose of
%JOSEP3. Reescrit.
the method described in this section, RR-Adjustment, is to
``repair'' some of the loss in estimation accuracy caused by 
these independence assumptions. The idea is to leverage 
the information about the relation between attributes
that remains in the randomized data set in order to obtain 
more accurate estimates of the joint distribution.

Our description of RR-Adjustment will be given in terms of attributes. 
However, the same description is valid if we substitute 
clusters of attributes for attributes. Indeed, 
a cluster of attributes can also be viewed as a special attribute 
obtained as the Cartesian product of the attributes in the cluster.
%JOSEP3. Afegit això
So {\em where we write ``attribute'',
``marginal distribution'' and ``RR-Independent'' in the rest of 
this section we could write  
``cluster of attributes'', 
``attribute cluster marginal distribution'' 
and ``RR-Clusters'', respectively.}

RR-Adjustment is based on the fact that, although attenuated,
the relation between attributes in ${\bf X}$ is likely to survive 
in ${\bf Y}$.
The greater the probability that randomization preserves the true 
values
({\em i.e.} the greater the probability mass in the diagonal
of the randomization matrix),
the better the relation between attributes is preserved in ${\bf Y}$.
For instance, when the RR matrix is the identity, ${\bf Y}$ contains the
same information as ${\bf X}$. In contrast, when randomization replaces 
each true value by a draw from a uniform distribution, no information remains in ${\bf Y}$.

%Jordi2020.
%JOSEP2020. Crec que aquest paràgraf no cal. Tot ja s'ha dit més amunt.
%Randomized response allows estimating the frequency
%of each attribute from ${\bf Y}$. However, it is not possible to estimate
% the joint distribution of several attributes; since attributes have been
%randomized independently, randomized response only allows estimating
%the marginal distribution of each attribute. To estimate
%the joint distribution, Protocol 1 assumes that attributes are independent.
%In this section, we avoid such a strong assumption by assuming that
%the relation between attributes is that shown by ${\bf Y}$. Whereas we do not
%claim that ${\bf Y}$ preserves the relations between attributes with complete
%accuracy, ${\bf Y}$ is likely to allow estimating those relations to   
%a reasonable extent.  
% This leads us to taking the empirical joint distribution
%of ${\bf Y}$ and altering its weights to match the estimated marginal distributions obtained from randomized response.

RR-Adjustment is formalized 
in Algorithm~\ref{alg:RR-adj}. The algorithm 
generates the estimate of the joint distribution 
of ${\bf X}$ by assigning weights (probabilities) to the records in ${\bf Y}$ in such a way that marginal
distributions coincide with the estimates  
obtained from RR-Independent. 
An iterative approach is used in this task. 
Initially, each point of ${\bf Y}$ is assigned 
 weight $1/n$. Then the algorithm loops through each attribute $A^j$ 
and adjusts the current weights
so that the marginal distribution of $A^j$ 
coincides with its RR-independent estimate. 
As the weight adjustment 
for $A_j$ is likely
to break the adjustments previously done,  
 the previous process needs to be iterated until 
it converges to a stable set of weights. If 
 strict convergence is desired,
iteration must carry on until weights do not change any more. However, 
less strict termination conditions such as
using a threshold on the weight changes or simply running a 
small fixed number of iterations 
are also valid. Recall that the relations 
between attributes in ${\bf Y}$   
have been attenuated by randomization; therefore, regardless of the 
convergence or termination criterion
used, the 
relations between attributes in ${\bf X}$ will be recovered only
approximately.

\begin{algorithm}
	\caption{RR-Adjustment. Adjustment of the randomized data set ${\bf Y}$ to match the estimated distributions
		of each attribute}
	\label{alg:RR-adj}
\begin{enumerate}
	\item Let ${\bf Y}=\{(y_i^1,\ldots,y_i^m)\}_{i=1,\ldots,n}$ be the randomized data set
	\item Let $\hat{\pi}^j$ be the estimated marginal distribution 
for attribute $X^j$
	\vspace{0.3cm}
%JOSEP3. Poso el vector de pesos per separat.
	\item Let ${\bf w}=(w_1, \ldots, w_n)$ be a vector
of weights for records in ${\bf Y}$
	\item Initialize $w_i=1/n$ for $i=1, \ldots, n$
	\item Repeat \{
	\item \hspace{0.3cm} For $j=1,\ldots,m$
	\item \hspace{0.6cm}    ${\bf Y'}$= Adjust\_weights (${\bf Y}$, 
${\bf w}$, $j$, $\hat{\pi}^j$)
	\item \} Until ($convergence\_or\_termination\_condition$)
 	\item Return ${\bf Y'}$
 	
 	\vspace{0.5cm}
 	\item Adjust\_weights(${\bf Y}$, ${\bf w}$, $j$, $\hat{\pi}^j$)
 	\item \hspace{0.3cm} For $k=1,\ldots,|X^j|$
 	\item \hspace{0.6cm}    Let $s_k= \sum_{y^j_i=k} w_i$ \\
//$s_k$ is the sum of weights of records in ${\bf Y}$ with $j$-th attribute
equal to $k$
 	\item \hspace{0.3cm} End for    
 	\item \hspace{0.3cm} For $i=1,\ldots,n$
	\item \hspace{0.6cm} 	$v= y^j_i$ 	
	\item \hspace{0.6cm}    Set $w_i = w_i \times \frac{\hat{\pi}^j_v}{s_v}$
	\item \hspace{0.3cm} End for
	\item \hspace{0.3cm} Return ${\bf Y}$, ${\bf w}$
	
\end{enumerate}
\end{algorithm}

%JOSEP4. Traslladat aquí.
Regarding privacy,
RR-Adjustment transforms the randomized data ${\bf Y}$ without making
use of the true data ${\bf X}$. Thus, RR-Adjustment does not
increase the risk of disclosure with respect to  ${\bf Y}$.

%Jordi202007.
As to the computational cost, $Adjust\_weights$ is called $m$ times for each iteration 
of the algorithm. Since each iteration of $Adjust\_weights$ has cost $\mathcal{O}(n)$, the total
cost of the algorithm is $\mathcal{O}(n\times m \times iter)$, where $iter$ is the number of iterations (this number depends on the termination criterion used).  

%JOSEP3. Ja d'ha dit més amunt. Ho trec.
%RR-Adjustment leverages the information about attribute
%dependences that remains in ${\bf Y}$ to estimate the
%joint distribution of ${\bf X}$ without assuming 
%independence between attributes. 
We next give a numerical toy example to 
illustrate the operation of Algorithm~\ref{alg:RR-adj}.

\begin{example}
%JOSEP3. Canviat l'inici per dir que és amb RR-Independent
{\em Consider a randomized data set ${\bf Y}$ 
obtained by using RR-Independent 
on $n=10$ parties with two attributes,
so that the first attribute can take values 
$a^1_1$ and $a^1_2$ and the second attribute
can take values $a^2_1$ and $a^2_2$.
The empirical joint distribution of ${\bf Y}$ is
as follows: 
\begin{eqnarray} 
&&\mbox{$(a^1_1, a^2_1)$ appears in the first 4 records,} \nonumber \\
&&\mbox{$(a^1_2,a^2_1)$ appears in the next 2 records,} \nonumber \\
&&\mbox{$(a^1_1,a^2_2)$ appears in 0 records and} \nonumber\\ 
&&\mbox{$(a^1_2,a^2_2)$ appears in the last 4 records.} \label{empiric}
\end{eqnarray}
This yields marginal distributions 
$\hat{\lambda}^1=(4/10, 6/10)$ and $\hat{\lambda}^2=(6/10,4/10)$
for the
two attributes $Y^1$, $Y^2$ in the randomized data set ${\bf Y}$.
%Jordi2020. After using 
Assume that after using 
Expression (\ref{unbiased}) independently on
each attribute, we obtain estimated
marginal distributions 
$\hat{\pi}^1=(1/2,1/2)$ and $\hat{\pi}^2=(1/2,1/2)$ for the two attributes
$X^1$, $X^2$ in the true data set ${\bf X}$.
Algorithm~\ref{alg:RR-adj} assigns an initial weight $w_i=1/10$ to each 
record for $i=1,\ldots,10$; 
adding these weights for records sharing the same value of $Y^j$ 
yields the marginal distribution of attribute $Y^j$. 
Then the algorithm adjusts the weights of records in order 
to make the marginal distributions of $Y^1$ and $Y^2$ as 
close as possible to the estimated marginal distributions
of $X^1$ and $X^2$:
\begin{itemize}
\item For the first attribute, the 
{\tt Adjust\_weights} routine in Algorithm~\ref{alg:RR-adj} 
changes $\hat{\lambda}^1$
from $(4/10, 6/10)$
into $(1/2, 1/2)$ (which is the value of $\hat{\pi}^1=(1/2,1/2)$.
To do this, the procedure computes $(s_1,s_2)=(4/10, 6/10)$.
Then for the first record, it sets $v=y^1_1=a^1_1$ and 
\[w_1=1/10 \times \frac{1/2}{4/10} = 1/8.\]
Similarly, for the second to fourth records $v=a^1_1$ and
$w_2=w_3=w_4=1/8$. 
Then for the fifth to tenth records $v=a^1_2$ and 
$w_5=w_6=w_7=w_8=w_9=w_{10}=1/12$.
After these changes, we have $\hat{\lambda}^1=(1/2,1/2)$ 
and, as a side effect, we also have an updated $\hat{\lambda}^2=(2/3,1/3)$
\item For the second attribute, {\tt Adjust\_weights} changes 
$\hat{\lambda}^2$
from $(2/3, 1/3)$
into $(1/2, 1/2)$ (which is the value of $\hat{\pi}^1=(1/2,1/2)$).
This is done in a way analogous to what was done for the first
attribute. As a side effect, this will change again $\hat{\lambda}^1$,
that will no longer be $(1/2,1/2)$. 
\end{itemize}
We see that changes in the distribution of one attribute result
in changes of the distribution of the other attribute. This is why
{\tt Adjust\_weights} must be iterated to try to bring the 
empirical distributions $\hat{\lambda}^1$, $\hat{\lambda}^2$ 
close to the estimated distributions $\hat{\pi}^1=(1/2,1/2)$, 
$\hat{\pi}^2=(1/2,1/2)$. In this example, this can be achieved
 because the joint empirical distribution converges towards
the first 4 records having weight $1/8$, the next 2 having weight 0,
and the last 4 records having weight $1/8$. 
Thus, RR-Adjust yields the following joint empirical distribution for the 
four combinations of categories:
\begin{eqnarray}
\label{dist1}
\Pr(a^1_1,a^2_1)=1/2; \Pr(a^1_1,a^2_2)=0;\nonumber\\ 
\Pr(a^1_2,a^2_1)=0; \Pr(a^1_2,a^2_2)=1/2.
\end{eqnarray}
In contrast, estimating the joint empirical distribution using 
just RR-Independent would yield:
\begin{eqnarray}
\label{dist2}
\Pr(a^1_1,a^2_1)=1/4; \Pr(a^1_1,a^2_2)=1/4;\nonumber \\
\Pr(a^1_2,a^2_1)=1/4; \Pr(a^1_2,a^2_2)=1/4.
\end{eqnarray}
In both cases, we have the same marginal distributions 
$\hat{\lambda}^1=(1/2,1/2)$ and $\hat{\lambda}^2=(1/2,1/2)$.
However, Distribution (\ref{dist1}) yielded by RR-Adjust 
is more similar than Distribution (\ref{dist2}) 
to the empirical distribution of 
the randomized data set ${\bf Y}$ (Expression (\ref{empiric})). 
Hence, Distribution (\ref{dist1}) is a more plausible 
estimation of the joint distribution of ${\bf X}$.}
\end{example}

\section{Experimental Results}
\label{sec:results}

\subsection{Dataset}
%Jordi2. Unknown --> Income
We based our experiments on the Adult dataset. This is a data set with over 
32,500 records and a combination of numerical and categorical attributes.
We assumed that each record was held by a different individual who
wanted to anonymize it locally with RR. 
For the experiments, we only took categorical attributes into account. These attributes were: Work-class (with 9 categories), Education (with 16 categories), Marital-status (with 7 categories), Occupation (with 15 categories), Relationship (with 6 categories), Race (with 5 categories), Sex (with 2 categories), and Income (with 2 categories). 

\subsection{Evaluated methods}
\label{metodes}

In the test dataset there were 1,814,400 possible combinations of attribute values. Such a large number 
made RR-Joint on the Cartesian product of all attributes 
computationally unfeasible. Even if RR-Joint had been 
computationally affordable, the estimated distribution 
would have been extremely inaccurate, because the number of categories
in the Cartesian product was much larger than the number of individual records 
(see Section~\ref{sub:accuracy}).

Discarding RR-Joint on all attributes for the above reasons 
left us with trying RR-Independent, RR-Clusters 
and RR-Adjustment to estimate the joint 
distribution of ${\bf X}$. We ran them as explained next:
\begin{enumerate}
	\item RR-Independent. We took as the baseline
for our experiments this method that performs 
RR independently for each attribute.
	\item RR-Clusters. We used RR-Clusters in an attempt to improve
	the estimation of the joint distribution with 
respect to RR-Independent. 
RR-Clusters was evaluated for different thresholds on the maximum number of
	category combinations in each cluster 
and for different thresholds on 
the minimum dependence required for attributes to be 
in the same cluster.
	\item RR-Independent + Adjustment. We leveraged 
the method described in Section~\ref{sec:Adjustment}
	to improve the distribution estimated with RR-Independent.
	\item RR-Clusters + Adjustment. We used 
the method of Section~\ref{sec:Adjustment} to improve 
the distribution estimated with RR-Clusters. 
\end{enumerate}

\subsection{Construction of the RR matrix}
\label{sub:rand_mat}

To make results across methods comparable, 
we evaluated the accuracy of the frequency estimates
at an equivalent level of risk. We used 
differential privacy as the risk measure.

First of all, we describe how we generated the RR matrix 
for RR-Independent. After that,
we describe how we generated an RR matrix 
for RR-Clusters that yielded an equivalent level of
differential privacy.

\subsubsection{RR matrix for RR-Independent}

Having selected differential privacy as the measure of risk, 
we wanted an RR matrix that was optimal
with respect to it. That is, we wanted an RR matrix that had 
the minimum level of randomization for
a given level of differential privacy. For an attribute $A$, such an RR matrix has the following form:
\begin{itemize}
	\item $p$ in the main diagonal;
	\item $(1-p)/|A|$ outside the main diagonal.
\end{itemize}

According to Expression (\ref{eqdp}), 
the level of differential privacy that such a matrix 
provided for attribute $A$ was 
%JOSEPNOV20. Posat valor absolut.
\[\epsilon_A = \left|\ln\left(\frac{p}{(1-p)/|A|}\right)\right|.
\]

\subsubsection{RR matrix for RR-Clusters}

RR-Clusters identifies clusters of attributes and applies 
RR-Joint {\em within} each cluster. Let $C$ be a cluster of attributes.
For the risk of RR-Clusters and RR-Independent to be equivalent, 
we needed RR-Clusters to yield
$\sum_{A\in C}\epsilon_A$-differential privacy on cluster $C$, where $\epsilon_A$ is the
level of differential privacy that RR-Independent yields for $A$.

We aimed for an RR matrix for cluster $C$ that is optimal at providing
$\sum_{A\in C}\epsilon_A$-differential privacy. Such a matrix had the form:
\begin{itemize}
	\item $p_C$ in the main diagonal;
	\item $p_C \exp {(-\sum_{A\in C}\epsilon_A)}$ outside the main diagonal.
\end{itemize}

To have a proper RR matrix, we needed 
the total probability mass in each row to be 1. This 
happened when
\[p_C = \frac{1}{1+(1-\Pi_{A\in C}|A|)\exp{(-\sum_{A\in C}\epsilon_A)}}.
\]

\subsection{Estimation of the true distribution}
\label{truedist}

The estimation returned by Equation (\ref{unbiased}) need not be a proper probability distribution.
In our experiments, we selected as an estimate 
of the true distribution the proper probability
distribution closest (according to the Euclidean distance) to 
the output
of Equation (\ref{unbiased}). 
Such distribution was found by applying the following procedure
to the output of Equation (\ref{unbiased}):
\begin{itemize}
	\item Replace any negative values by 0;
	\item Rescale the rest of values so that their sum is 1.
\end{itemize}

\subsection{Evaluation results}

In line with the measures used in Section~\ref{sub:accuracy}, we measured 
the accuracy
of the estimated distribution of the true data 
as the absolute and the relative error in count queries. 

%JOSEP. Negreta S
Let ${\bf S}$ be a subset of the data domain. Let $X_{\bf S}$ be
the true number of records of ${\bf X}$ that belong to ${\bf S}$. 
Let $Y_{\bf S}$ be the
number of records in ${\bf S}$ estimated from the randomized 
data set ${\bf Y}$.
The absolute error is
\[e_{\bf S} = |Y_{\bf S} - X_{\bf S}|
\]
and the relative error is
\begin{equation}
\label{relatiu}
r_{\bf S} = \left|\frac{Y_{\bf S}-X_{\bf S}}{X_{\bf S}}\right|.
\end{equation}

In all subsequent experiments, the values reported for $e_{\bf S}$ and 
$r_{\bf S}$ are median values over 1000 runs.

The subsets ${\bf S}$ of the data domain that we used in this evaluation 
were generated as follows:
\begin{itemize}
\item We chose a proportion $\sigma$ of the data domain 
 that we wanted ${\bf S}$ to cover.
	\item We took two random attributes of Adult to define 
 ${\bf S}$ (the results with ${\bf S}$ configured by a higher number of attributes
did not differ significantly).
	\item We randomly chose the subset ${\bf S}$ to contain 
a $\sigma$ proportion of all the possible combinations of values
of the previously selected two attributes.
\end{itemize}

First, we measured the accuracy of the baseline method, RR-Independent. 
We also measured the accuracy of the empirical 
distribution of the raw randomized data ${\bf Y}$
obtained with RR-Independent ({\em without} using Expression (\ref{unbiased})
to estimate the true distribution); we called this distribution
``Randomized''.

Figure~\ref{fig:base} shows the absolute and
the relative error in counts for $p=0.7$, as a function 
of the coverage $\sigma$, both for RR-Independent
and for Randomized. 
We observe that, thanks to 
Expression (\ref{unbiased}), RR-Independent
significantly reduced the absolute and the relative errors 
with respect to Randomized. 
The absolute error peaked at $\sigma=0.5$: the more combinations
in ${\bf S}$, the larger the absolute error could be. For larger
$\sigma$, the absolute error decreased: 
for example, for $\sigma=0.6$ the error was the same as for $\sigma=0.4$
(the absolute error of taking 60\% of categories is the same 
as the absolute error associated with the remaining 40\% of 
categories that are not covered). 
Regarding the relative error, it decreased as $\sigma$ grew,
because $X_{\bf S}$ in the denominator of Expression (\ref{relatiu}) was 
larger. 

\begin{figure}
	\begin{centering}
		\includegraphics[bb=50bp 225bp 540bp 580bp,clip,width=0.5\columnwidth]{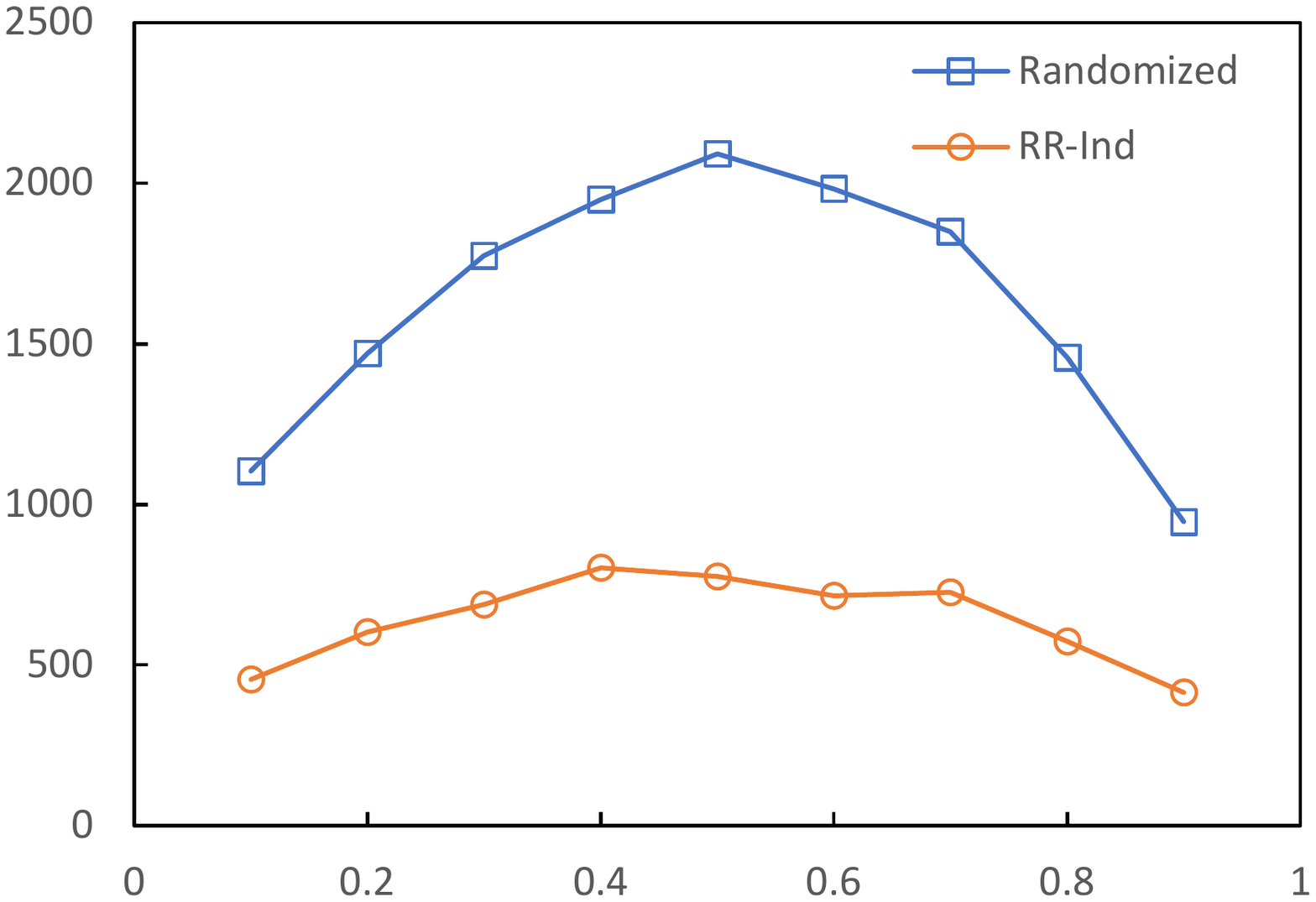}\includegraphics[bb=50bp 225bp 540bp 580bp,clip,width=0.5\columnwidth]{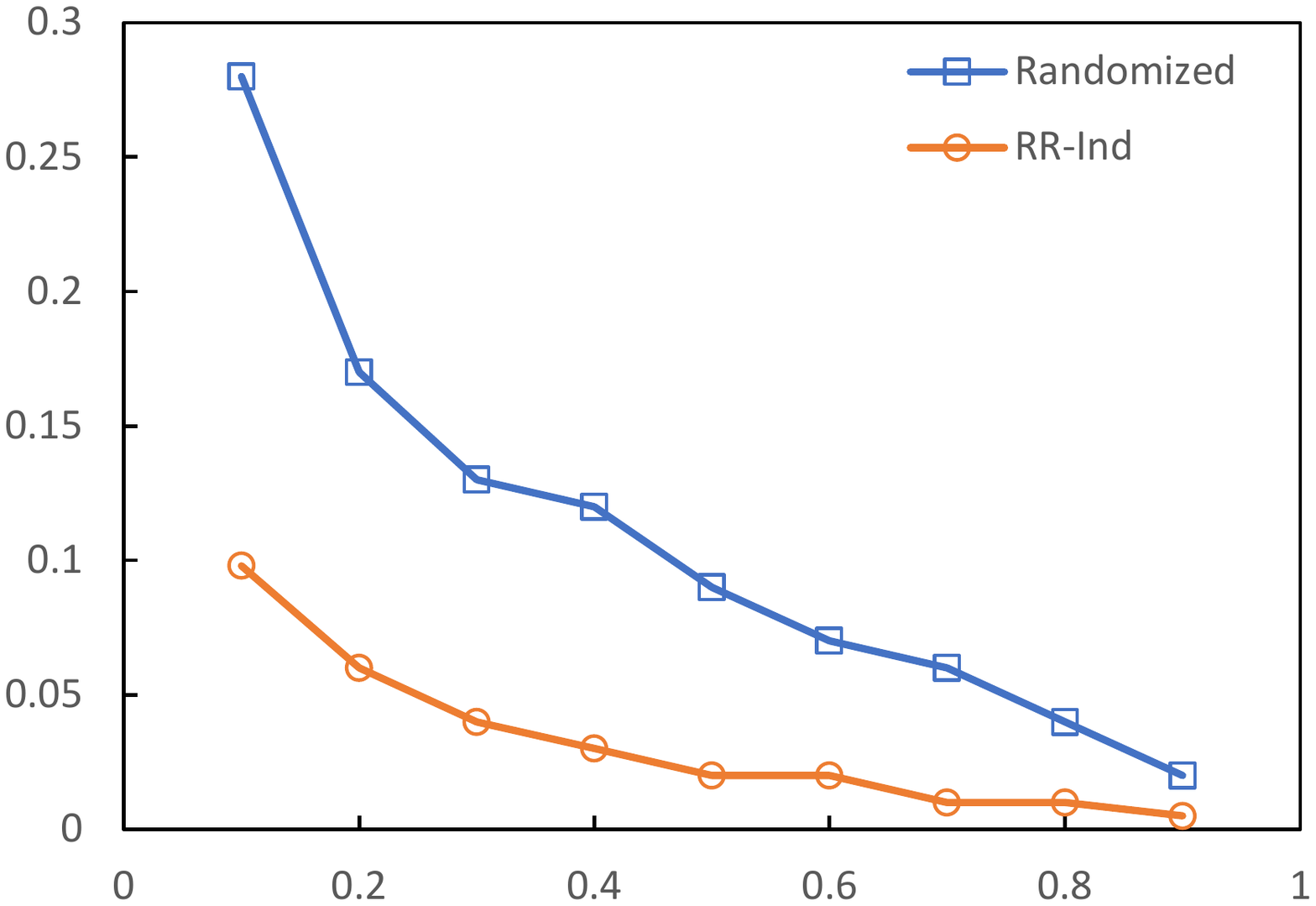}
		\par\end{centering}
	\caption{Absolute error (left) and relative error (right) of Randomized
		and RR-Independent when $p=0.7$. The $x$-axis represents 
the domain 
coverage $\sigma$.\label{fig:base}}
\end{figure}

The accuracy of estimations reported by RR-Clusters depends on 
the thresholds $T_v$ (maximum number of category 
combinations per cluster) and $T_d$ (minimum required dependence
between attributes in a cluster), 
the actual data set and the randomization matrix used. 
Next, we analyzed the behavior of RR-Clusters for the Adult data set 
and for randomization matrices generated as per Section~\ref{sub:rand_mat}.
In this evaluation, ${\bf S}$ was generated with $\sigma=0.1$.
Table~\ref{tab:threshold} shows the relative error of RR-Clusters for $T_v\in \{50,100,300\}$, $T_d\in \{0.1,0.2,0.3\}$ and a randomization matrix with $p\in \{0.1,0.3, 0.5, 0.7\}$. We observe that:
\begin{itemize}
\item As a rule, the relative error 
increased with $T_v$.  
This means that clusters with a high number of category combinations 
 had a clear negative effect on the estimation accuracy.
\item Regarding $T_d$, for small $p$ taking larger $T_d$ 
yielded better accuracy, whereas for larger $p$ taking smaller 
$T_d$ yielded better accuracy.  
Note that $T_d=0$ means that attributes can be clustered regardless of
their dependence, whereas $T_d=1$ 
means that attributes are never clustered (in this case we have
RR-Independent). 
Thus, clustering attributes turned out to be 
more rewarding for larger $p$, whereas
for small $p$, there was less incentive for clustering and the advantage
of RR-Clusters on RR-Independent was less noticeable.
\end{itemize}

\begin{table}
	
	\caption{Relative error of RR-Clusters in the Adult data set for $T_v\in \{50,100,300\}$, $T_d\in \{0.1,0.2,0.3\}$ and a randomization matrix with $p\in \{0.1,0.3, 0.5, 0.7\}$ }
	\label{tab:threshold}
	\begin{centering}
			\begin{tabular}{cc|ccc}
				\hline 
				    &         &    &$T_{v}$&\tabularnewline
				$p$ & $T_{d}$ & 50 & 100 & 300\tabularnewline
				\hline 
				0.1 & 0.1 & 0.335 & 0.404 & 0.495\tabularnewline
				0.1 & 0.2 & 0.357 & 0.351 & 0.501\tabularnewline
				0.1 & 0.3 & 0.285 & 0.426 & 0.505\tabularnewline
				\hline 
				0.3 & 0.1 & 0.335 & 0.334 & 0.426\tabularnewline
				0.3 & 0.2 & 0.262 & 0.310 & 0.435\tabularnewline
				0.3 & 0.3 & 0.199 & 0.306 & 0.445\tabularnewline
				\hline 
				0.5 & 0.1 & 0.094 & 0.148 & 0.214\tabularnewline
				0.5 & 0.2 & 0.107 & 0.127 & 0.236\tabularnewline
				0.5 & 0.3 & 0.116 & 0.119 & 0.212\tabularnewline
				\hline 
				0.7 & 0.1 & 0.069 & 0.069 & 0.074\tabularnewline
				0.7 & 0.2 & 0.070 & 0.075 & 0.071\tabularnewline
				0.7 & 0.3 & 0.070 & 0.068 & 0.079\tabularnewline
				\hline 
			\end{tabular}
		\par\end{centering}
\end{table}

Next, we compared the accuracy of the methods listed 
in Section~\ref{metodes} for different levels of
randomization ($p\in \{0.1, 0.3, 0.5, 0.7\}$) and for different coverages 
of ${\bf S}$ ($\sigma \in \{0.1, 0.2, 0.3, 0.4,$ $0.5, 0.6, 0.7,$ $0.8, 0.9\})$.
Figure~\ref{fig:rel_error} shows the relative error in the count queries. 
For each $p$, we took 
%JOSEP4. IMPORTANT. T_d and T_d corregit a T_v and T_d
the best values for $T_v$ and $T_d$ identified in 
Table~\ref{tab:threshold}.  
We observe that:
\begin{itemize}
\item For small values of $p$ (that is, for
$p=0.1$ and $p=0.3$), RR-Independent yielded the best accuracy. 
For these values, using RR-Clusters
and RR-Adjustment was counter-productive. 
\item For larger values of $p$ (that is, for $p=0.5$
and $p=0.7$) and large coverages ($\sigma \geq 0.3$)
all methods behaved similarly and offered a very small relative error.
\item For larger $p$ and {\em small} coverages ($\sigma < 0.3$), 
RR-Clusters offered much more accuracy than RR-Independent. Furthermore,
using RR-Adjustment brought a substantial accuracy
improvement, no matter whether plugged after RR-In\-de\-pend\-ent or RR-Clusters.
\end{itemize}

In summary, for strong randomization, the dependences 
between attributes in ${\bf X}$ are mostly lost in ${\bf Y}$.
For this reason, in this case RR-Independent is as good
as RR-Clusters, and RR-Adjustment does not bring much. 
In contrast, for weak randomization, it makes sense
to leverage whatever dependences might be preserved 
in ${\bf Y}$, and thus RR-Clusters and RR-Adjustment 
outperform RR-Independent. This superior behavior
is visible only for small coverages, because for large
coverages the denominator of Expression (\ref{relatiu}) is
so large that any method achieves a small relative error.
 
% Preview source code for paragraph 1

\begin{figure}
	\begin{centering}
		\includegraphics[bb=40bp 225bp 560bp 580bp,clip,width=0.5\columnwidth]{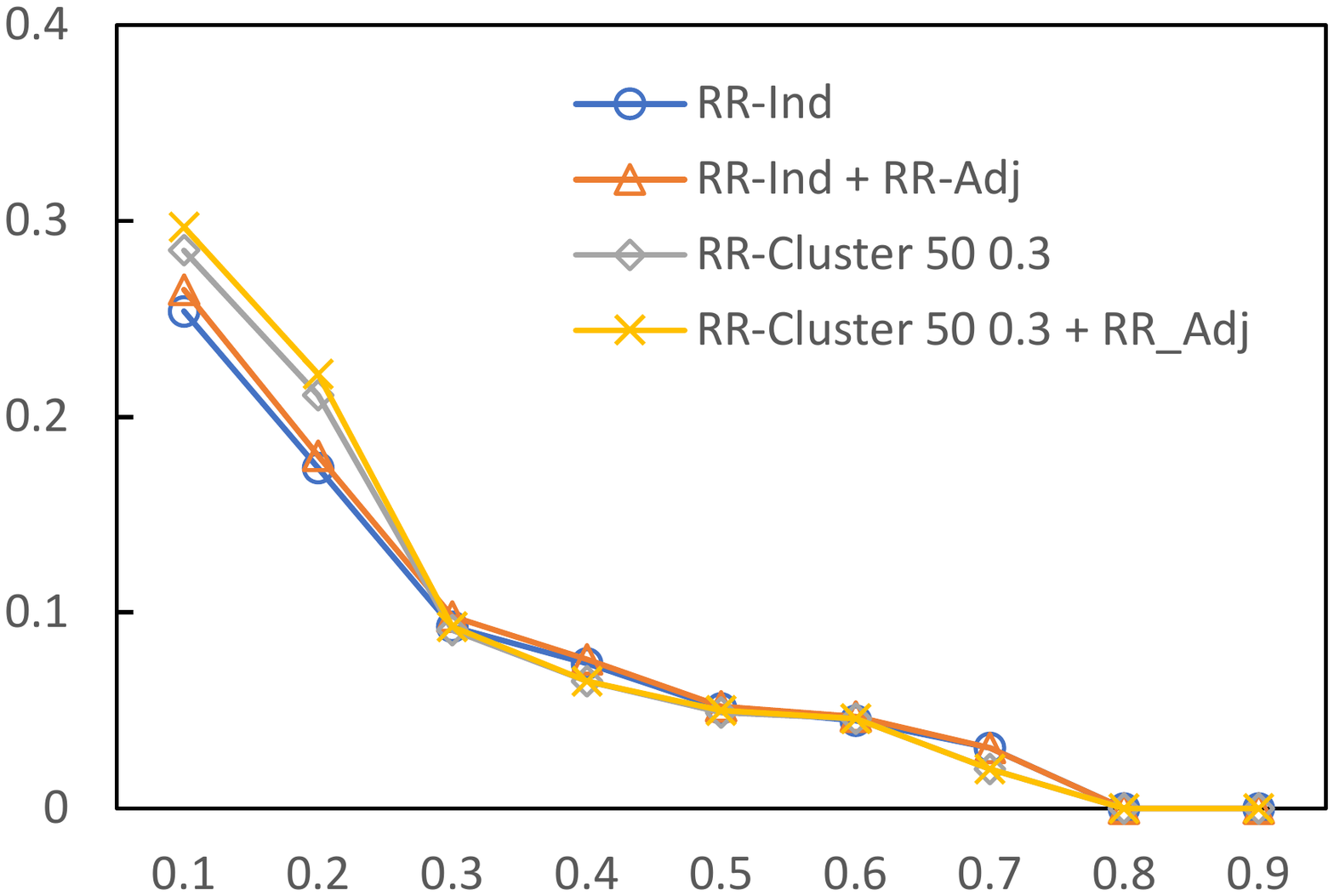}\includegraphics[bb=40bp 225bp 560bp 580bp,clip,width=0.5\columnwidth]{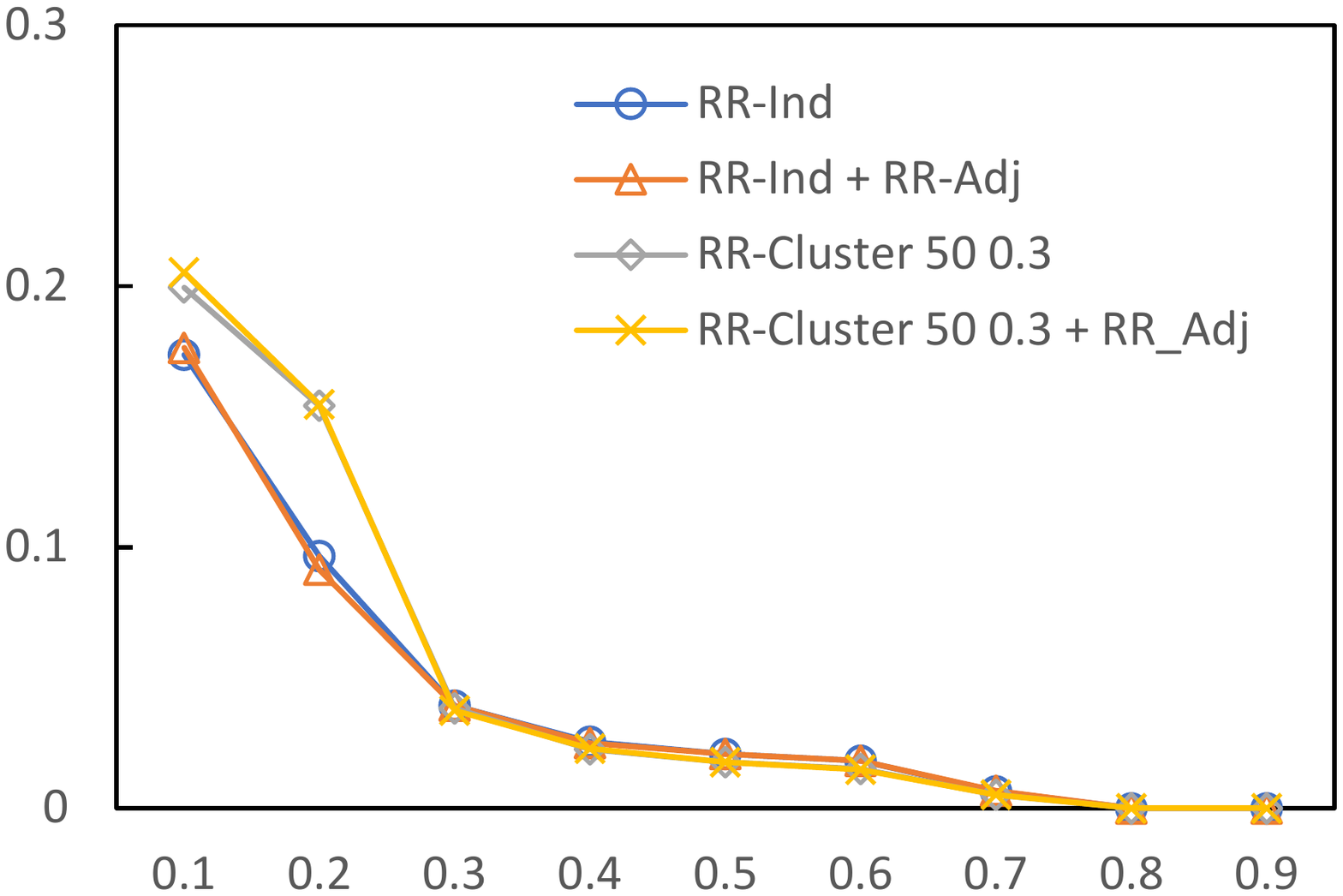}
		\par\end{centering}
	\begin{centering}
		\includegraphics[bb=40bp 225bp 560bp 580bp,clip,width=0.5\columnwidth]{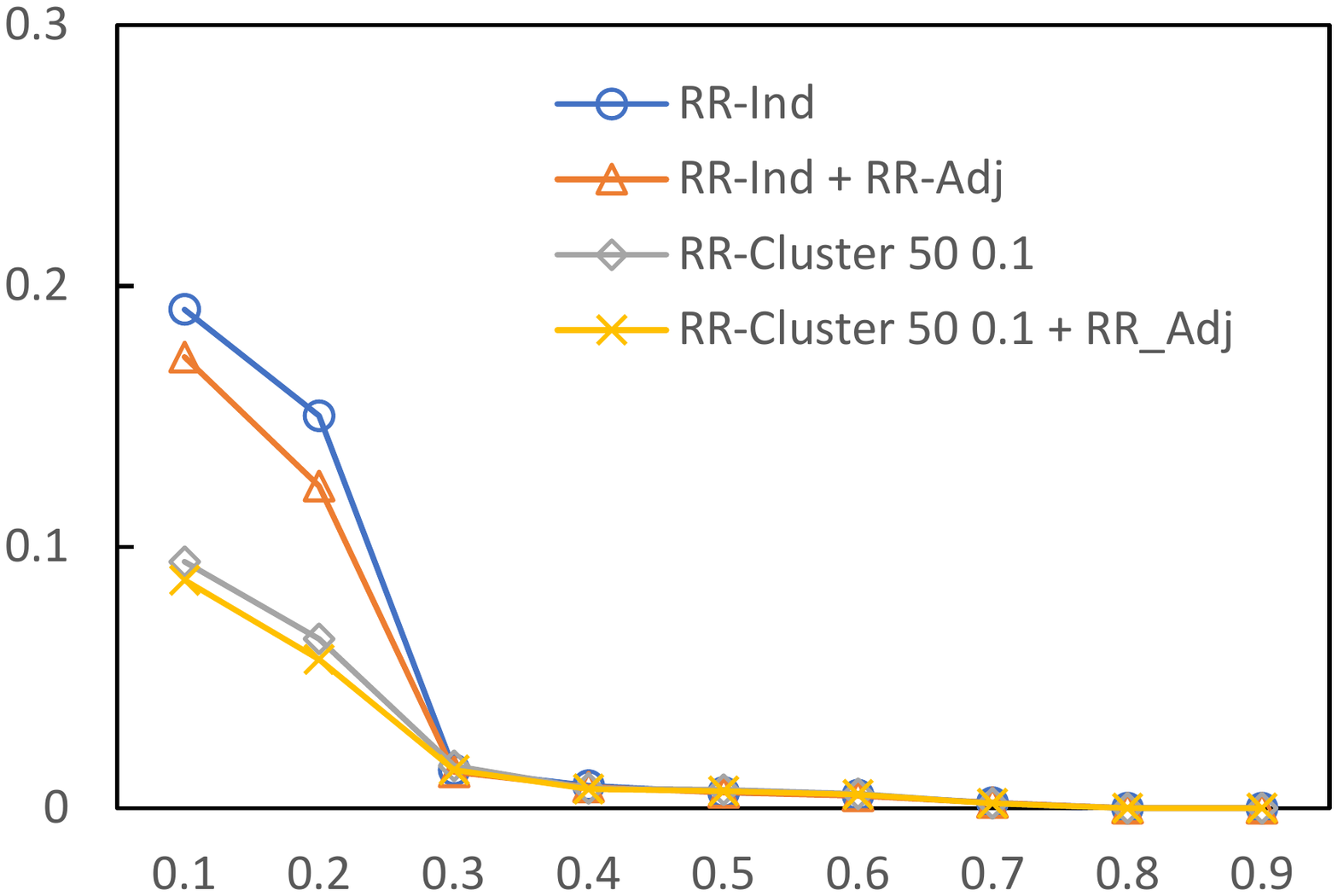}\includegraphics[bb=40bp 225bp 560bp 580bp,clip,width=0.5\columnwidth]{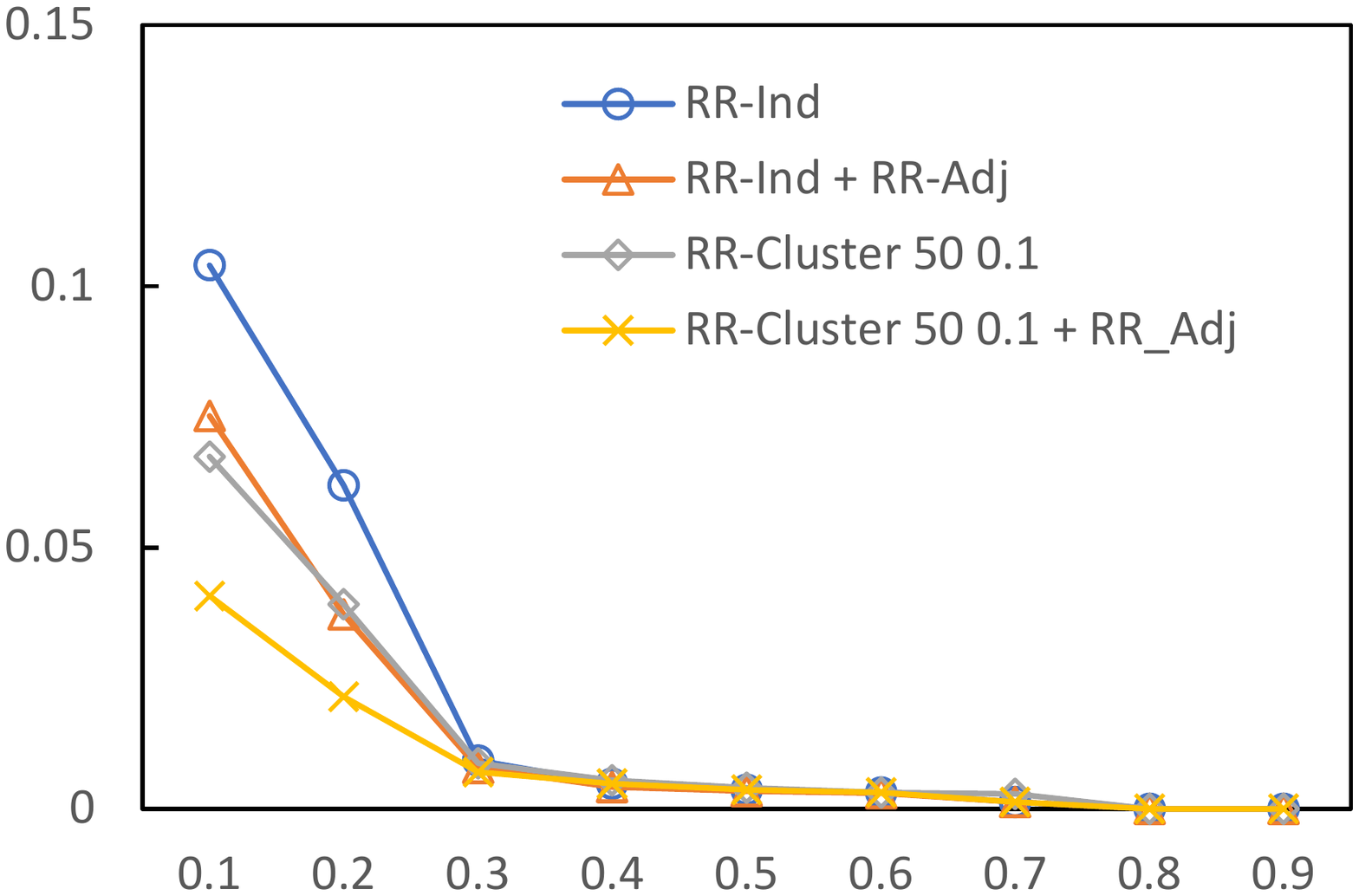}
		\par\end{centering}
	\caption{Relative error for different randomizations: 
 $p=0.1$ (top left),
		$p=0.3$ (top right), $p=0.5$ (bottom left) and 
		 $p=0.7$ (bottom right). 
The $x$-axis represents the domain coverage $\sigma$.
\label{fig:rel_error}}
	
\end{figure}

Finally, we analyzed the effect of the data set size 
on the accuracy of the estimates.
For the comparison with the previous results to be fair,  
we needed a data set with
the same distribution. We obtained an expanded data set Adult6 
by concatenating the original 
Adult data set 6 times. In this way, Table~\ref{tab:threshold} 
for Adult became
Table~\ref{tab:threshold-2} for Adult6, which shows the relative 
error of RR-Clusters for $T_v\in \{50,100,300\}$, 
$T_d\in \{0.1,0.2,0.3\}$ and a randomization 
matrix with $p\in \{0.1, 0.3, 0.5, 0.7\}$. By comparing Tables~\ref{tab:threshold} and~\ref{tab:threshold-2}
we observe that the relative error decreased 
for all parameterizations but the reduction achieved depended
on the specific parameterization. Specifically:
\begin{itemize}
\item For $p=0.7$, the reduction was small,
because for so little randomization 
the relative error was already quite low in Table~\ref{tab:threshold}.
In this case, the highest reduction occurred for $T_v=300$, because
Adult6 being larger than Adult, a larger number of category
combinations had a lower negative impact on estimation accuracy.
This highest reduction for $T_v=300$ caused the lowest 
relative error for $p=0.7$ to occur when $T_v=300$, which
shows the advantages of allowing a great number of category
combinations when the data set is sufficiently large.
\item For smaller $p$ (that is, for higher randomization levels),
the reduction in the relative error was more remarkable 
for $T_v=50$ and $T_v=100$. These thresholds on category combinations
could be better accommodated using the larger data set Adult6. 
In contrast, allowing up to $T_v=300$ category combinations 
had an impact on the relative error that was compensated
only partially by the increase in the data set size; hence,
the reduction in the relative error for $p=0.1,0.3, 0.5$
and $T_v=300$ was less spectacular than for lower $T_v$.
Further, unlike for $p=0.7$, the highest $T_v$ did not 
achieve the lowest relative error: this seems to indicate
that, as the randomization level increases, one needs 
larger data set sizes to be able to work with a great
number of category combinations.
\item The effect of the threshold $T_d$ on dependence 
did not seem to change with the data set size.  
\end{itemize}

\begin{table}
	\caption{Relative error of RR-Clusters in the 
Adult6 data set for $T_v\in \{50,100,300\}$, $T_d\in \{0.1,0.2,0.3\}$ and a randomization matrix with $p\in \{0.1, 0.3, 0.5, 0.7\}$}
	\label{tab:threshold-2}
	\centering{}%
	\begin{tabular}{cc|ccc}
			\hline 
			&  &  & $T_{v}$  & \tabularnewline
		 	$p$ & $T_d$ & 50 & 100 & 300\tabularnewline 
			\hline 
                        0.1 & 0.1 & 0.189 & 0.312 & 0.459\tabularnewline
                        0.1 & 0.2 & 0.173 & 0.310 & 0.449\tabularnewline
                        0.1 & 0.3 & 0.183 & 0.339 & 0.462\tabularnewline
                        \hline 
                        0.3 & 0.1 & 0.149 & 0.202 & 0.369\tabularnewline
                        0.3 & 0.2 & 0.171 & 0.225 & 0.376\tabularnewline
                        0.3 & 0.3 & 0.178 & 0.217 & 0.369\tabularnewline
			\hline 
			0.5 & 0.1 & 0.080 & 0.084 & 0.123\tabularnewline
			0.5 & 0.2 & 0.082 & 0.075 & 0.126\tabularnewline
			0.5 & 0.3 & 0.083 & 0.079 & 0.127\tabularnewline
			\hline 
			0.7 & 0.1& 0.064 & 0.066 & 0.056\tabularnewline
			0.7 & 0.2 & 0.064 & 0.066 & 0.057\tabularnewline
			0.7 & 0.3 & 0.065 & 0.065 & 0.060\tabularnewline
			\hline 
	\end{tabular}
\end{table}

\section{Related work}
\label{related}

%Jordi3.
%JOSEP4. Lleugerament reescrit.
Privacy preservation in data set releases has a substantial 
tradition in the statistics
and computer science communities. Privacy models such as $k$-anonymity~\cite{Samarati},
$t$-closeness~\cite{Li} and differential privacy~\cite{Dwork2006}, as well as 
many statistical disclosure control techniques~\cite{Hundepool}, have been used 
to protect data sets before releasing them. All these works assume there is
a trusted party that collects the true original data and takes care of protecting them.

Some attempts to enforce privacy models in a distributed manner 
have been made. For example,~\cite{Soria} 
proposed a way to enforce $k$-anonymity by promoting the collaboration between users.
Also, local differential privacy has been proposed
as an adaptation of differential privacy to the untrusted collector
scenario.
 Most of the work in local differential privacy targets the distributed computation of 
data analytics. However, some well-known attempts to generate locally differentially private data sets
have been made. For instance, RAPPOR~\cite{rappor} generates a data set that allows users to compute 
the frequencies of a given set of items. RAPPOR is based on randomized response and Bloom filters.
An extension of RAPPOR is available that allows users to compute multivariate distributions~\cite{Fanti}.
While privacy models usually offer privacy guarantees at the record level ({\em e.g.} $k$-anonymity limits
the chances of succesful re-identification of a record based on quasi-identifier attributes), attaining such
guarantees becomes increasingly difficult as the number of attributes grows. This difficulty is serious in 
the trusted data collector scenario, and {\em a fortiori} in the 
more complex untrusted data collector scenario.

Some masking techniques can be applied locally by each individual before 
releasing her record 
({\em e.g.} noise addition or generalization). However, the lack 
of a global view generally prevents adjusting the
masking to the data set formed by all original individual records ({\em e.g.} 
to preserve covariances between attributes or to apply a
stronger generalization to those values that are rare). 
In this context, randomized response is very convenient because  
despite being a local masking approach it allows 
the data collector to estimate the distribution of the
original data. However, estimating the true original 
joint distribution is only
feasible for a small number of attributes. 
%JOSEP2020. Afegit.
For example, in~\cite{Wang16} an approach is presented that 
starts with RR followed by estimation of the original
distribution for a single binary attribute, then generalizes
to a single multicategory attribute, and finally to several
multicategory attributes (the case we deal with in this paper).
However, the authors of~\cite{Wang16} apply RR independently 
for each attribute, 
in a way similar to our RR-Independent protocol, which does not preserve
the relations between the original attributes. 
Our proposal is to strike a compromise between dimensionality
reduction and preservation of attribute relations by performing
RR on clusters of attributes.
%Jordi3

In~\cite{oganpsd2018}, a method was proposed for clustering the attributes
of high-dimensional data sets in view of mitigating the curse
of dimensionality. The authors used hierarchical clustering algorithms
%JOSEP4. algorithms -> attributes
to identify clusters of strongly correlated attributes. A substantial 
difference with our paper is that they considered the centralized anonymization
paradigm, in which a data controller holds the entire original data set
and can compute attribute dependences in a straightforward way. 
In contrast, we deal with local anonymization, in which the original
record of each individual is only known to that individual.

An attribute clustering approach to allow randomized response 
of several attributes was presented in~\cite{psd2018}. 
Unfortunately, this method requires non-negligible 
disclosure of original attributes to compute attribute dependences.
Furthermore, it can yield large 
attribute clusters with too many category
combinations and/or nearly independent attributes, 
which hampers accuracy.
Our attribute
clustering method is superior in several aspects: 
i) it does not cluster attributes unless 
the number of category combinations of the clustered attributes
is below a certain threshold and unless the dependence between
the clustered attributes is above a certain threshold; 
ii) it specifies three carefully crafted procedures to compute
attribute dependences with minimum privacy
loss for the individual parties; iii) our adjustment algorithm
can partly compensate the accuracy lost by assuming independence
between attributes in different clusters.

\section{Conclusions and future research}
\label{conc}

Randomized response is an appealing anonymization approach 
in our big data era for 
several reasons. On the one side, it offers local anonymization
and on the other side it can yield microdata useful
for exploratory analysis and even machine learning. 
The main hindrance for using RR on multivariate data is the curse
of dimensionality: as the number of attributes grows, 
the accuracy of the estimated distribution for the true original
microdata quickly degrades. 

In this paper, we have proposed mitigations to 
the dimensionality problem, based on performing RR
separately for each attribute ---which implicitly
assumes all attributes are (nearly) independent--- 
or jointly within clusters of attributes ---which
needs attributes in different clusters to be weakly
dependent if not independent. 
We have then proposed a method to recover some 
of the estimation accuracy loss incurred by
the above independence assumptions.

The proposed approaches open future research avenues.
Randomized response assumes that all attributes are
categorical or can be made categorical. Thus, a challenge 
is to devise
local anonymization approaches yielding good distribution 
estimates for original
{\em numerical} microdata with a large number
of attributes. Another intriguing issue is
whether there exist alternative ways to recover a larger share
of the utility loss incurred by independence assumptions.
Yet more daunting is to tackle a scenario in which 
all attributes are so correlated that any independence assumption
will result in unaffordable accuracy loss.

\section*{Acknowledgment and Disclaimer}

%JOSEP2020. Canviats projectes als acks.
Thanks go to Rafael Mulero-Vellido for his help in the empirical work.
We are also indebted to Oriol Farr\`as for help with the secure sum
protocol.
Partial support to this work has been received from the European Commission
(project H2020-871042 ``SoBigData++''),
the Government of Catalonia (ICREA Acadèmia Prize to J. Domingo-Ferrer
and grant 2017 SGR 705),
and the Spanish Government (project RTI2018-095094-B-C21 ``CONSENT''). The 
first author is with the UNESCO Chair in
Data Privacy, but the views in this paper are the authors' own and are not
necessarily shared by UNESCO.

\bibliographystyle{acm}

\begin{IEEEbiography}[{\includegraphics[width=1in,height=1.25in,clip,keepaspectratio]{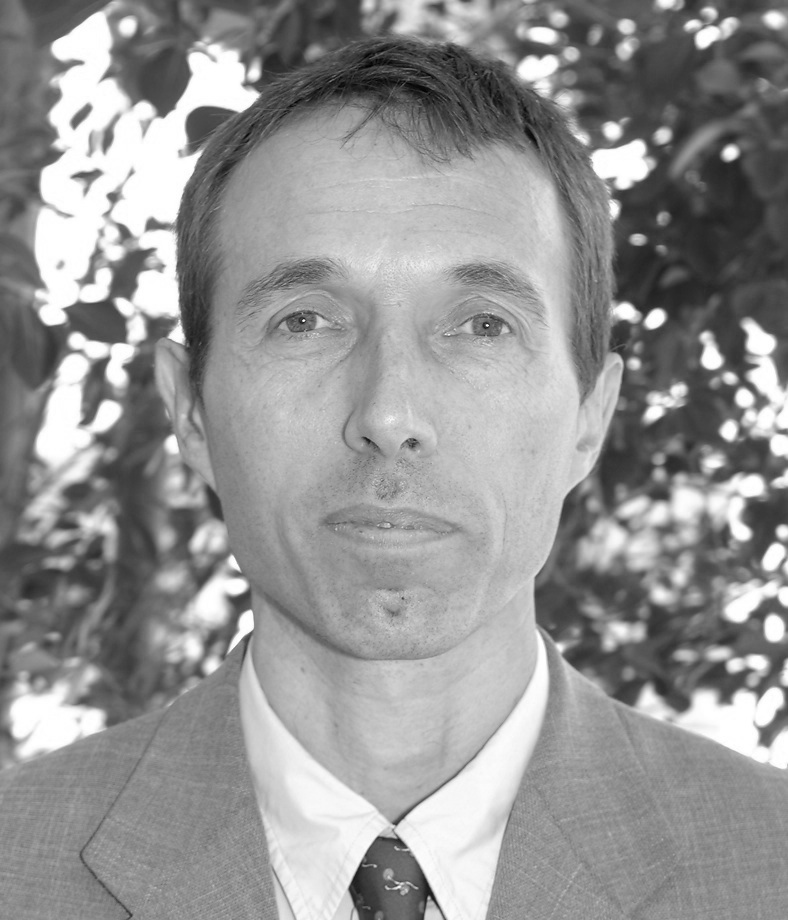}}]{Josep Domingo-Ferrer}
(Fellow, IEEE) is a distinguished professor of computer science and
an ICREA-Acad\`emia researcher at Universitat Rovira i Virgili, Tarragona,
Catalonia, where he holds the UNESCO Chair in Data Privacy and leads
CYBERCAT. He received
the MSc and PhD degrees in computer science from the Autonomous University
of Barcelona in 1988 and 1991, respectively. He also holds an MSc
degree in mathematics. His research interests are in data privacy,
data security and cryptographic protocols.
\end{IEEEbiography}

\begin{IEEEbiography}[{\includegraphics[width=1in,height=1.25in,clip,keepaspectratio]{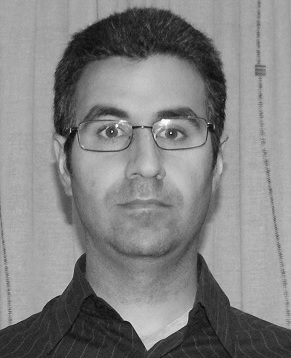}}]{Jordi Soria-Comas}
is the Co-ordinator of Technology and Information Security
at the Catalan Data Protection Authority, in Barcelona. 
He received
an MSc in computer security (2011) and a PhD in computer science (2013)
from Universitat Rovira i Virgili. He also holds an MSc in finance
from the Autonomous University of Barcelona (2004) and a BSc in mathematics
from the University of Barcelona (2003). His research interests are
in data privacy and security.
\end{IEEEbiography}

\end{document}